\documentclass[11pt]{article}

\usepackage[margin=1in]{geometry}
\usepackage{amsmath,amsthm,amsfonts,amssymb,dsfont}
\usepackage{enumitem,framed,comment}

\allowdisplaybreaks[1]

\providecommand{\noqed}{\sbox{0}{\popQED}}
\newcommand{\maxstep}{max-step}

\usepackage{cite}
\usepackage{privbib2}
\usepackage{hyperref}

\usepackage[nokwfunc,noresetcount]{algorithm2e} % add option vlined
\DontPrintSemicolon
\SetAlgoSkip{normalskip}
\newcommand{\IlIf}[2]{\KwSty{if} #1 #2}

\newcommand{\Class}[1]{{\AlCapSty{\small Object} \FuncSty{#1}}}
\newcommand{\Method}[1]{\caption{\AlCapSty{\small Method} \FuncSty{#1}}}

\newlength{\myalgowidth}
\newenvironment{classfigure}[1][htb]{
  \begin{figure}[#1]\begin{minipage}{\textwidth}\centering%
        \setcounter{AlgoLine}{0}\NoCaptionOfAlgo\LinesNotNumbered%
      }{\end{minipage}\end{figure}}

\makeatletter
\newenvironment{class}[1]{
  \RestyleAlgo{plain}\@twocolumnfalse\begin{algorithm}[H]\Class{#1}\;
}{\end{algorithm}\@twocolumntrue}

\newenvironment{method}[1]{%
  \vskip1ex\hfill%
  \centering%
  \begin{minipage}{\myalgowidth}%
  \LinesNumbered%
  \@twocolumnfalse
  \RestyleAlgo{boxruled}\begin{algorithm}[H]\Method{#1}
}{\end{algorithm}\@twocolumntrue\end{minipage}}
\makeatother

\SetKwFunction{win}{win}
\SetKwFunction{lose}{lose}
\SetKwFunction{falloff}{fall-off}
\SetKwFunction{enter}{enter}
\SetKwFunction{heads}{heads}
\SetKwFunction{tails}{tails}
\SetKwFunction{True}{true}
\SetKwFunction{False}{false}
\SetKwFunction{Read}{read}
\SetKwFunction{Write}{write}
\SetKwFunction{Left}{left}
\SetKwFunction{Right}{right}
\SetKwFunction{Stop}{stop}
\SetKwFunction{Doorway}{Doorway}
\SetKwFunction{Splitter}{Splitter}
\SetKwFunction{RSplitter}{RSplitter}
\SetKwFunction{split}{split}
\SetKwFunction{op}{op}

\SetKwFunction{TAS}{TAS}
\SetKwFunction{TASII}{TAS$_2$}
\SetKwFunction{LEII}{LeaderElect$_2$}
\SetKwFunction{LeaderElect}{LeaderElect}
\SetKwFunction{GroupElect}{GroupElect}
\SetKwFunction{elect}{elect}
\SetKwFunction{RatRace}{RatRace}
\SetKwFunction{Imp}{Imp}
\SetKwFunction{Comb}{Comb}
\SetKwFunction{mComb}{Comb$'$}

\SetKwInput{shared}{shared}
\SetKwInput{local}{local}
\SetKw{myID}{myID}

\newtheorem{theorem}{Theorem}%[section]
\newtheorem{lemma}[theorem]{Lemma}
\newtheorem{corollary}[theorem]{Corollary}

\renewcommand{\AA}{\mathcal{A}}
\newcommand{\DD}{\mathcal{D}}

\newcommand{\EE}{\mathcal{E}}

\newcommand{\ceil}[1]{\lceil#1\rceil}
\newcommand{\floor}[1]{\lfloor#1\rfloor}
\newcommand{\Exp}{\ensuremath{\operatorname{\mathbf{E}}}}

\newcommand{\IIR}{\mathds{R}}

\newcommand{\Up}{\mathit{Up}}
\newcommand{\Down}{\mathit{Down}}

\title{Efficient Randomized Test-And-Set Implementations}

\author{
George Giakkoupis\\
INRIA, Rennes, France\\
george.giakkoupis@inria.fr
\and
\and
Philipp Woelfel\\
University of Calgary\\
woelfel@ucalgary.ca
}

\begin{document} %%%%%%%%%%%%%%%%%%%%%%%%%%%%%%%%%%%%%%%%%%%%%%%%%%

\maketitle
\sloppy

\begin{abstract}
We study randomized test-and-set (TAS) implementations from registers in the asynchronous shared memory model with $n$ processes.
We introduce the problem of \emph{group election}, a natural variant of leader election, and propose a framework for the implementation of TAS objects from group election objects.
We then present two group election algorithms, each yielding an efficient TAS implementation.
The first implementation has expected \maxstep{} complexity $O(\log^\ast k)$
%\todo{use just step complexity in the abstract?}
in the location-oblivious adversary model, and the second has expected \maxstep{} complexity $O(\log\log k)$ against any read/write-oblivious adversary, where $k\leq n$ is the %point
contention.
These algorithms improve the previous upper bound by Alistarh and Aspnes~\cite{AA_TAS_2011a} of $O(\log\log n)$ expected \maxstep{} complexity in the oblivious adversary model.

We also propose a modification to a TAS algorithm by Alistarh, Attiya, Gilbert, Giurgiu, and Guerraoui~\cite{AAGGG2010a} for the strong adaptive adversary, which improves its space complexity from super-linear to linear, while maintaining its $O(\log n)$ expected \maxstep{} complexity.
We then describe how this algorithm can be combined with any randomized TAS algorithm that has expected \maxstep{} complexity $T(n)$ in a weaker adversary model, so that the resulting algorithm has $O(\log n)$ expected \maxstep{} complexity against any strong adaptive adversary and $O(T(n))$ in the weaker adversary model.

Finally, we prove that for any randomized 2-process TAS algorithm, there exists a schedule determined by an oblivious adversary such that with probability at least $1/4^t$ one of the processes needs at least $t$ steps to finish its TAS operation.
This complements a lower bound by Attiya and Censor-Hillel~\cite{Attiya2010-consensus} on a similar problem for $n\geq 3$ processes.
\end{abstract}

%\category{D.1.3}{Programming Techniques}{Concurrent Programming}[Distributed programming]
%\category{F.2.2}{Analysis of Algorithms and Problem Complexity}{Nonnumerical Algorithms and Problems}
%
%
%%\terms{Algorithms, Theory}
%
%\keywords{Test-And-Set, leader election, shared memory, randomization, strong/weak adversary, time/space complexity}

\section{Introduction}
\label{sec:into}

In this paper we study time and space efficient implementations of \emph{test-and-set  (TAS)} objects from atomic registers in asynchronous shared memory systems with $n$ processes.
The TAS object is a fundamental synchronization primitive, and has been used in algorithms for classical problems such as mutual exclusion and renaming~\cite{KRS1988a,PPTV1997a,EHW1998a,BPSV2006a,AAGGG2010a,AAGG2011a,AACGZ2011a}.

A TAS object stores a bit that is initially 0, and supports the operation \TAS{}, which sets the bit (or leaves it unchanged if it is already set) and returns its previous value; the process whose call returns 0 is the \emph{winner} of the object.
TAS objects are among the simplest natural primitives that have no deterministic wait-free linearizable implementations from atomic registers, even in systems with only two processes.
In fact, in systems with exactly two processes, a consensus protocol can be implemented deterministically from a TAS object and vice versa.
% In systems with $n\geq 3$ processes consensus objects are more powerful, as no deterministic wait-free implementation of a consensus protocol from TAS objects and atomic registers exists.

The TAS problem is very similar to the problem of \emph{leader election}.
In a leader election protocol, every process decides for itself whether it becomes the leader (it returns \win) or whether it loses (it returns \lose).
At most one process can become the leader, and not all participating processes can lose.
I.e., if all participating processes finish the protocol, then exactly one of them returns \win and all others return \lose.
Obviously, any TAS object immediately yields a leader election protocol:
Each process executes a single \TAS{} operation and returns \win if the \TAS{} call returns 0, or \lose if \TAS{} returns 1.
Similarly, a leader election algorithm, together with one additional register, can be used to implement a linearizable TAS object with just a constant increase in the number of steps \cite{Art-GHW2010a}.
Similar transformations from leader election to linearizable TAS objects are implicit in several TAS algorithms, e.g., \cite{Afek1992,AA_TAS_2011a}.

Early randomized TAS implementations assumed a \emph{strong adaptive adversary} model, where the adversary bases its scheduling decisions on the entire past history of events, including the coin flips by processes.
%\pwtodo{There are two problems with the term ``adaptive''. 1) the other adversaries are also ``adaptive'' in the literal meaning. 2) it is confusing because we also talk about adaptive \emph{algorithms}. In very early papers, I think the ``adaptive'' adversary has been called ``strong''. Some authors have used ``strong adaptive''. I suggest to use ``strong adaptive adversary''. This would fix at least 1) and mitigate 2) somewhat.}
Tromp, and Vit\'{a}nyi~\cite{TV1990a,Tromp2002} presented a randomized implementation for two processes which has constant expected \maxstep{} complexity and constant space complexity against any strong adaptive adversary.
(The \maxstep{} complexity of an execution is the maximum number of steps any process needs to finish its algorithm in the execution.
%and the \emph{expected} \maxstep{} complexity of an algorithm against some adversary $A$ is the expectation of the \maxstep{} complexity of the random execution obtained by that algorithm if processes are scheduled by $A$.
See Section~\ref{sec:prelim-complexity} for formal definitions and a discussion.)
%(The \emph{expected \maxstep{} complexity} is the expectation of the maximum number of steps any process needs to finish the algorithm during an execution.
Afek, Gafni, Tromp, and Vit\'{a}nyi~\cite{Afek1992} gave a deterministic implementation of a TAS object for $n$ processes, from $O(n)$ 2-process TAS objects.
Any execution of this algorithm has \maxstep{} complexity $O(\log n)$.
Using Tromp and Vit\'{a}nyi's randomized 2-process TAS implementation, one obtains a randomized implementation of a TAS object from registers with $O(\log n)$ expected \maxstep{} complexity in the strong adaptive adversary model.
Alistarh, Attiya, Gilbert, Giurgiu, and Guerraoui~\cite{AAGGG2010a}
presented an \emph{adaptive} variant of that algorithm, called \RatRace, in which the expected \maxstep{} complexity is logarithmic in the contention $k$, i.e., the total number of processes accessing the TAS object.
The space requirements of \RatRace are higher, though, as $\Theta(n^3)$ registers are used.

%\pwtodo{Cite $O(\sqrt n)$ TAS paper?}
%\gtodo{I think it would be easier if we moved this paragraph to the end of the section where we discuss our DISC'13 result}
%%As any other wait-free algorithm, it is possible to transform \RatRace into an algorithm that uses only $n$ multi-reader single-writer registers of large enough size:
%%Multi-reader multi-writer registers can be simulated by multi-reader single-writer registers~\cite{LTB1996a}, and it is easy to see that for any algorithm that uses multi-reader single-writer registers, $n$ such registers are enough.
%%However, this straightforward transformation yields a polynomial blowup in time complexity.
%It is currently unknown whether any randomized wait-free (or obstruction-free) TAS implementation exists that uses fewer than $O(n)$ registers.
%Styer and Peterson~\cite{Styer1989} proved a lower bound which shows that any implementation of deadlock-free leader election requires at least $\lceil\log n\rceil+1$ registers.
%They also described a deadlock-free (deterministic) leader election algorithm that uses $\lceil\log n\rceil+1$ registers.
%However, this algorithm is not wait-free (and thus has unbounded step complexity).
%
No TAS algorithm with a sub-logarithmic expected \maxstep{} complexity against any strong adaptive adversary has been found yet, and no non-trivial time lower bounds are known either.
The strong adaptive adversary, however, may be too strong in some settings to model realistic system behavior.
Motivated by the fact that consensus algorithms benefit from weaker adversary models, Alistarh and Aspnes~\cite{AA_TAS_2011a} devised a simple and elegant TAS algorithm with an expected \maxstep{} complexity of $O(\log\log n)$ for the \emph{oblivious adversary} model, where the adversary has to make all scheduling decisions at the beginning of the execution.
We will refer to this algorithm as the \emph{AA-algorithm}.
Although not explicitly mentioned in~\cite{AA_TAS_2011a}, the AA-algorithm works even for a slightly stronger adversary, the \emph{read/write-oblivious (r/w-oblivious)} adversary.
Such an adversary can take all past operations of processes, including coin flips, into account when making scheduling decisions, but it cannot see whether a process will read or write in its next step, if that decision is made by the process at random.
The space complexity of the AA-algorithm is super-linear, as it uses \RatRace as a component.

\begin{table*} [tb]
\begin{center}
    \begin{tabular}{| c | c | c | c | l |}
    \hline
    \textbf{Adversary}
    &
    \textbf{Time}
    &
    \textbf{Space}
    &
    \textbf{Reference}
    &
    \textbf{Comments}
    \\
    \hline
    \hline
    strong adaptive
    &
    unbounded
    &
    $\lceil\log n\rceil+1$
    &
    \cite{Styer1989}
    &
    deadlock-free only
    \\
    \hline
    strong adaptive
    &
    $O(1)$
    &
    $\Theta(1)$
    &
    \cite{TV1990a,Tromp2002}
    &
    2-process implementation
    \\
    \hline
    strong adaptive
    &
    $O(\log n)$
    &
    $\Theta(n)$
    &
    \cite{Afek1992}
    &
    -
    \\
    \hline
    strong adaptive
    &
    $O(\log k)$
    &
    $\Theta(n^3)$
    &
    \cite{AAGGG2010a}
    &
    $k$ is the contention
    \\
    \hline
    r/w-oblivious
    &
    $O(\log\log n)$
    &
    $\Theta(n^3)$
    &
    \cite{AA_TAS_2011a}
    &
    -
    \\
    \hline
    location-oblivious
    &
    $O(\log^\ast k)$
    &
    $\Theta(n)$
    &
    Theorem~\ref{thm:location-oblivious}
    &
    -
    \\
    \hline
    r/w-oblivious
    &
    $O(\log\log k)$
    &
    $\Theta(n)$
    &
    Theorem~\ref{thm:rw-oblivious}
    &
    -
    \\
    \hline
    strong adaptive
    &
    $O(\log k)$
    &
    $\Theta(n)$
    &
    Theorem~\ref{thm:ratrace}
    &
    -
    \\
    \hline
    oblivious
    &
    $O(\log^\ast k)$
    &
    $\Theta(\log n)$ %$\Theta(\sqrt n)$
    &
    \cite{GHHW15a} %\cite{GHHW2013a}
    &
    uses impl.\ of Theorem~\ref{thm:location-oblivious}
    \\
    \hline
    \end{tabular}
    \caption{Randomized TAS implementations. In the second column  we give the expected \maxstep{} complexity of the algorithm.
%    \gtodoi{not sure what to put in this table; fix the caption; add ref to the table in the main text}
    }
    \label{tab:}
\end{center}
\end{table*}

\paragraph{Our contribution.}
In view of their AA-algorithm, Alistarh and Aspnes asked whether any better TAS algorithm exists for the oblivious or even stronger adversary models.
We answer this question in the affirmative: We present an adaptive algorithm that has an expected \maxstep{} complexity of $O(\log^\ast k)$ in the oblivious adversary model, where $k$ is the
%maximum point
contention.
In fact, our result holds for the slightly stronger \emph{location-oblivious} adversary.
%{\todo{Need a new name.}}
This adversary makes scheduling decisions based on all past events (including coin flips), but it does not know which register a process will access in its next step, if this decisions is made at random.

This algorithm, however, is not efficient in the r/w-oblivious adversary model.
For such adversaries, we devise a different algorithm that has expected \maxstep{} complexity $O(\log\log k)$, and uses $O(n)$ registers.
%It is based on the ``sifting'' idea of the AA-algorithm, but applies a new twist to it in order to make it adaptive.
It is similar to the AA-algorithm, but introduces a new idea that makes it adaptive.

Our two TAS algorithms above are the first ones with sub-logarithmic expected \maxstep{} complexity that need only $O(n)$ registers.

Both algorithms rely on a novel framework that uses  a variant of the leader election problem, called \emph{group election}, in which more than one process can get elected.
We present a TAS implementation based on multiple such group election objects.
The performance of the implementation is determined by the \emph{effectiveness} of the group election objects used, which is measured in terms of the expected number of processes that get elected.

The AA-algorithm has the desirable property that its performance degrades gracefully when the adversary is not r/w-oblivious, and against a strong adaptive adversary it still achieves an expected \maxstep{} complexity of $O(\log k)$.
%(Recall that it is currently unknown whether there is any TAS algorithm with sub-logarithmic expected \maxstep{} complexity against the adaptive adversary.)
%This is a desirable property, as one does not need to rely on the system not to behave like an adaptive adversary.
In their basic form, our algorithms do not exhibit such a behavior---a strong adaptive adversary can find a schedule where processes need $\Omega(k)$ steps to complete their \TAS{} operation.
To rectify that, we present a general method to combine any TAS algorithm with \RatRace, so that
%the combined algorithm inherits the \maxstep{} complexity against any of the two adversaries from the corresponding algorithm for the same \todo{Check this sentence} adversary.
% its \maxstep{} complexity from both adversary models.
%I.e.,
if the algorithm has expected \maxstep{} complexity $T(k)$ against any r/w-oblivious or location-oblivious adversary, then the combined algorithm has  expected \maxstep{} complexity $O\bigl(T(k)\bigr)$ in the same adversary model, and $O(\log k)$ against any strong adaptive adversary.
Further, we propose a modification of \RatRace that improves its space complexity from $O(n^3)$ to $O(n)$, without increasing its expected \maxstep{} complexity.
Thus, combining this algorithm with any of our two algorithms for weak adversaries, yields an algorithm with linear space complexity.

% We complement our algorithms with lower bounds.
% First, we show that at least $\Omega(\log n)$ registers are needed for any randomized TAS implementation (from atomic registers) that satisfies the progress condition \emph{nondeterministic solo-termination}~\cite{Fich1998}, which is strictly weaker than wait-freedom.
% This is the first non-trivial lower bound on the space complexity of randomized TAS implementations.

Finally, we show for any randomized TAS implementation for two processes, that the oblivious adversary can schedule processes in such a way that for any $t>0$, with probability at least $1/4^t$
one of the processes needs at least $t$ steps to finish its \TAS{} operation.
%one of the processes does not finish its \TAS{} method call within fewer than $t$ steps.
This result immediately implies the same lower bound on 2-process consensus.
%complementing a lower bound by Attiya and Censor-Hillel~\cite{Attiya2010-consensus}, who showed
Attiya and Censor-Hillel~\cite{Attiya2010-consensus} showed
that with probability at least $1/c^{t}$, for some constant $c$, any randomized $f$-resilient $n$-process consensus algorithm does not terminate within a total number of $t(n-f)$ steps.
However, the lower bound proof in~\cite{Attiya2010-consensus} only works for $n\geq 3$ processes.
Thus our result fills in the missing case of $n=2$.

In the conference version of this paper~\cite{GW2012b}, we also proved a lower bound of $\Omega(\log n)$ for the number of registers needed to implement nondeterministic solo-terminating TAS.
After the conference paper was published, Dan Alistarh made us aware that a proof by Styer and Peterson from 1989~\cite{Styer1989} implies this result.
In particular, Styer and Peterson~\cite{Styer1989} showed that any implementation of deadlock-free leader election requires at least $\lceil\log n\rceil+1$ registers.
They also described a deadlock-free (deterministic) leader election algorithm that uses $\lceil\log n\rceil+1$ registers.
However, this algorithm is not wait-free (and thus has unbounded step complexity).

Until recently, it was not unknown whether any randomized wait-free (or obstruction-free) TAS implementation exists that uses fewer than $O(n)$ registers.
After completion of the draft of this paper, Giakkoupis, Helmi, Higham, and Woelfel~\cite{GHHW2013a,GHHW15a} presented deterministic obstruction-free algorithms that use only $O(\sqrt n)$ and $O(\log n)$ registers, respectively.
As the authors observed, these algorithms can be turned into randomize wait-free ones, and can be combined with the first algorithm proposed in this paper to achieve $O(\log^\ast n)$ expected \maxstep{} complexity in the oblivious adversary model, with $O(\sqrt n)$ and $O(\log n)$ space complexity, respectively.

\section{Preliminaries}
\label{sec:prelims}

We consider an asynchronous shared memory model where up to $n$ processes, with IDs $1,\dots,n$, communicate by reading and writing to atomic shared multi-reader multi-writer registers.
Registers can store values from an arbitrary countable domain.
Algorithms are randomized and use local coin flips to make random decisions.
A coin flip is a step that yields a random value from some countable space $\Omega$, using an arbitrary but fixed probability distribution $\DD$. %\pwtodo{Check!}
Coin flips are private, i.e., only the process that executes the coin flip gets to see the outcome.
For the model description we will assume (w.l.o.g.) that processes alternate between coin flip steps and shared memory steps (i.e., reads or writes), and that their first step is always a coin flip.
Our algorithm descriptions do not always follow this convention, because in the given programs processes may execute multiple consecutive shared memory steps without any coin flips in-between.
Obviously one can simply add ``dummy'' coin flip steps in order to achieve an alternation.

An \emph{execution} is a possibly infinite sequence, where the $i$-th element contains all information describing the $i$-th step.
That comprises the ID of the process taking that step, the type of step (read, write, or coin flip), the affected register in case of a read or write, the value returned in case of a read or coin flip, and the value written in case of a write.
A \emph{schedule} is a sequence of process IDs in $\{1,\ldots,n\}$, and a \emph{coin flip vector} $\omega$ is a sequence of coin flip values in $\Omega$; these sequences may be infinite.
Every execution $\EE$ uniquely defines a schedule $\sigma(\EE)$ that is obtained from $\EE$ by replacing each step with the ID of the process performing that step, and a coin flip sequence $\omega(\EE)$, which is the sequence of coin flip values defined by $\EE$.
Similarly, for a given algorithm $M$, a schedule $\sigma$ together with an infinite coin flip vector $\omega=(\omega_1,\omega_2,\dots)$ uniquely determine an \emph{execution} $\EE_M(\sigma,\omega)$, in which processes execute their shared memory and coin flip steps in the order specified by $\sigma$, and the value returned from the $i$-th coin flip (among all processes) is $\omega_i$.
If a process has finished its algorithm, it does not take any more steps, even if it gets scheduled
(alternatively, one can think of the process continuing to execute only no-ops).

%\gtodo{elsewhere we use the term ``dummy steps"; let's stick to one of the two}
%\pwtodo{There is a difference between no-ops and dummy steps: We use no-ops to say that an execution is well-defined, even if a terminated process gets scheduled (i.e., the process just doesn't do anything). In Section 5, where we use dummy steps, those are actual steps, and they get counted in the complexity.}

\subsection{Adversary Models}
\label{sec:prelim-adversary}

An adversary decides at any point of an execution, which process will take the next step.
Formally, an \emph{adversary} $A$ is a function that maps a finite execution $\EE$ of some algorithm $M$ to a process ID $A(\EE)$, which identifies the process to take the next step following $\EE$.
This way, adversary~$A$ and algorithm~$M$, together with an infinite coin flip vector $\omega$, yield a unique infinite schedule $\sigma_M(A,\omega)=(\sigma_1,\sigma_2,\dots)$,
%\pwtodo{Maybe $\sigma(M,A,\omega)$ is better.}
%\gtodo{I think it is ok}
where $\sigma_1=A(\varepsilon)$ for the empty execution $\varepsilon$, and
\begin{displaymath}
  \sigma_{i+1}=A\Bigl(\EE_M\bigl((\sigma_1,\dots,\sigma_i),\omega\bigr)\Bigr).
\end{displaymath}
Thus, given algorithm $M$ and adversary $A$ we can obtain a random schedule $\sigma_M(A,\omega)$ and the corresponding random execution $\EE_M(\sigma_M(A,\omega),\omega)$ by choosing a coin flip vector $\omega$ at random
according to the product distribution $\DD^\infty$ over the set $\Omega^\infty$ of infinite coin flip vectors.
The coin flip vector $\omega$ is the only source of randomness, here.
We denote the random execution $\EE_M(\sigma_M(A,\omega),\omega)$ by $\EE_{M,A}$, and call it the \emph{random execution of $M$ scheduled by $A$}.
% or simply by $\EE$ if $M$ and $A$ are clear from the context.
We are interested in random variables and their expectation defined by $\EE_{M,A}$, e.g., the maximum number of shared memory steps any process takes (see Section~\ref{sec:prelim-complexity}).

An \emph{adversary model $\AA$} maps each algorithm $M$ to a family $\AA(M)$ of adversaries.
We %sometimes
say that an algorithm $M$ has certain properties
against any adversary in $\AA$
%\emph{against} a certain type of adversary
to denote that these properties are satisfied for any adversary
%chosen from the corresponding set $\AA(M)$.
$A\in \AA(M)$.
The \emph{strong adaptive adversary model} is defined for any algorithm as the set of all adversaries.
Here, the next process scheduled to take a step is decided based on the entire past execution (including the results of all coin flip steps so far).
The \emph{oblivious adversary model} is the weakest standard adversary model,
%The weakest standard adversary model is the \emph{oblivious} one,
where each adversary $A$ is a function of just the length of the past execution, i.e., $A(\EE) = A(\EE')$, if $|\EE|=|\EE'|$.
Therefore, an oblivious adversary results in a schedule that is fixed in advance and is independent of the coin flip vector.

Several \emph{weak}  adaptive adversary models have been proposed, which are stronger than the oblivious model but weaker than the strong adaptive model.
We will consider two such models.
%An adversary $A_M$ for algorithm $M$ is \emph{location-oblivious} if $A_M(\EE)$ can depend on the following information:
An adversary $A$
for algorithm $M$
is \emph{location-oblivious} if for any finite execution $\EE$ of $M$, the next processes $A(\EE)$ scheduled by $A$ to take a step can depend on the following information:
\begin{enumerate}[label=(\roman*)]
  \item the complete past schedule $\sigma(\EE)$;
  \item the return values of all coin flip steps performed by each process $p$ \emph{preceding} $p$'s latest shared memory step in $\EE$; and
  \item
  % OLD:
%  for each process $p$ that does not finish in $\EE$, the information whether $p$'s next step will be a read, write, or coin flip step, and, in case of a  write, the value that $p$ will write.
%    \gtodoi{%
%    \gtodo{I slightly modified (iii) and (iii'), and the sentences right after, based on the model assumption of alternating coin flip / shared memory steps.}
    for each process $p$ that does not finish in $\EE$ and its next step is a shared memory step, the information whether that step will be a read or a write operation, and, in case of a  write, the value that $p$ will write.
%    (Recall the model assumption of alternating coin flip / shared memory steps.)
%    for each process $p$ the information whether or not $p$ finishes in $\EE$; also if $p$ does not finish and its next step is a shared memory step, the information whether that step will be a read or write, and, in case of a  write, the value that $p$ will write.
%    }
\end{enumerate}
%Note that $p$'s next step is uniquely determined by algorithm $M$ and execution $\EE$, except for the value that will be returned from a read or coin flip value.
In particular, the location-oblivious adversary does not make a scheduling decision based on \emph{which register} each process $p$ will access in its next shared memory step, if that register is determined at random based $p$'s coin flip \emph{after} its latest shared memory step in $\EE$.
%\pwtodo{I removed the following to avoid confusion: ``
%A stronger variant of this model is the one in which the adversary knows also the register that a process will read from, if its next step is a read; but it does not know the register it will write on in case of a write.
%The algorithm presented in Section~\ref{sec:location-algo}, works in fact even for this stronger adversary model.'' If we want it, we should move it to Section~\ref{sec:location-algo}.}

Similar but incomparable to the location-oblivious adversary model is the \emph{r/w-oblivious} adversary model.
%In this model, any function $A_M(\EE)$ for algorithm $M$ and execution $\EE$ can depend on information (i) and (ii) above, and also on the following information:
An adversary $A$
for algorithm $M$
is r/w-oblivious if for any finite execution $\EE$ of $M$, $A(\EE)$ can depend on (i) and (ii) above, and also on the following information:
\begin{enumerate}%[label=(\roman*)]
  \item[(iii$'$)]
  % OLD:
%    for each process $p$ that does not finish in $\EE$, the information whether $p$'s next step will be a shared memory step or a coin flip, and, in case of a shared memory step, the register that $p$ will access in that step.

    for each process $p$ that does not finish in $\EE$ and its next step is a shared memory step, the register that $p$ will access in that step.
%  \pwtodo{I removed ``and it knows which register $p$ will access next, and also it knows the value that $p$ will write to that register if $p$'s next step is a write.'' This makes the adversary slightly stronger, but it is confusing.}
\end{enumerate}
%In particular, the adversary's scheduling decisions must be independent of whether a process $p$'s next shared memory step is a read or a write operation, if this decision is made at random based on coin flips executed after $p$'s last shared memory step in $\EE$.
In particular, the adversary does not make a scheduling decision based on
whether a process $p$'s next shared memory step is a read or a write operation, if this decision is made at random based on $p$'s coin flip after its last shared memory step in $\EE$.

%For our results it does not matter whether adversaries make random or deterministic choices.
%In fact, we can always assume the model that yields the stronger result:
%The analysis of our algorithms apply to randomized adversaries, and our lower bound in Section~\ref{sec:lb-two-processes} uses a deterministic adversary.

\subsection{Complexity Measures}
\label{sec:prelim-complexity}

We use the following standard definitions.
The \emph{space complexity} of an implementation is the number of registers it uses.
An event occurs \emph{with high probability (w.h.p.)}, if it has probability ${1-1/m^{\Omega(1)}}$ for some parameter $m$, as $m\to \infty$.
In our case, $m$ will be either $n$, the total number of processes, or $k$, a notion of congestion defined in Section~\ref{sec:prelim-complexity}.

We are interested in randomized leader election, and a variant of it called group election.
These problems are \emph{one-time} in the sense that each process can participate in a leader (or group) election at most once.
The following definitions are thus limited to one-time operations~$op$.

Let $M$ be an algorithm in which a processes may call some operation $op$
%at most once
(possibly in addition to other operations).
For any process $p$ and any execution $\EE$ of algorithm $M$, let $T_{op,p}(\EE)$ be the number of shared memory steps that $p$ executes in $\EE$ during its $op$ call, and let $T_{op,p}(\EE)=0$ if $p$ does not call $op$.
The \emph{\maxstep{} complexity} of $op$ in execution $\EE$ is defined as
\[
    \max_{p} T_{op,p}(\EE).
\]
The \emph{expected \maxstep{} complexity} of $op$ in algorithm $M$ against an adversary $A$ is
\begin{equation}
    \label{eq:maxTopp}
    \Exp\left[\max_{p} T_{op,p}(\EE_{M,A})\right],
\end{equation}
where $\EE_{M,A}$ is a random execution of $M$ scheduled by $A$ (see Section~\ref{sec:prelim-adversary}).
The expected \maxstep{} complexity of $op$ against $A$ is the supremum of the quantity in~\eqref{eq:maxTopp} over all algorithms $M$. %in which processes invokes $op$,
The expected \maxstep{} complexity of $op$ against an adversary model $\AA$ is the supremum of the quantity in~\eqref{eq:maxTopp} over all $M$ and all $A\in \AA(M)$.

In previous works~\cite{AAGGG2010a,AA_TAS_2011a}, the terms ``expected individual step complexity'' or simply ``expected step complexity'' have been used to denote what we refer to as ``expected \maxstep{} complexity.''
We prefer to use a new and thus unambiguous term to clearly distinguish this measure from other step complexity measures, and in particular, from $\max_{p}\Exp[T_{op,p}(\EE_{M,A})]$.
It follows immediately from the definition of expectation that $\max_{p}\Exp[T_{op,p}(\EE_{M,A})]\leq \Exp[\max_{p}T_{op,p}(\EE_{M,A})]$.

Our  implementations of group and leader election
%\gtodo{same comment as above}
objects are \emph{adaptive} with respect to contention, i.e., their \maxstep{} complexity depends on the number of participating processes rather than $n$, the number of processes in the system.
In fact, the only way in which $n$ is used in the design of our algorithms is to determine the number of registers that must be used.
If we allow the implementation to use unbounded space, then $n$ can be unbounded, too.

Expressing the \maxstep{} complexity in terms of contention requires some care.
We are interested in the conditional expectation of the \maxstep{} complexity of an operation $op$, given that the number of processes calling $op$ is limited by some value $k$.
A straightforward idea
to limit contention would be
to consider
$
%\begin{equation*}
    %\label{eq:adaptive_max_step_complexity_intuition}
  \Exp [\max_{p} T_{op,p}(\EE_{M,A}) \mid  K \leq k ],
%\end{equation*}
$
where $K$ is the actual number of processes that execute $op$ in $\EE_{M,A}$.
But this does not yield satisfying results, as an adaptive adversary may be able to force that conditional expectation to be unreasonably large for any given $k<n$.
An adversary might achieve that, e.g., by letting $k$ processes start their operation $op$, and if it sees during the execution that the coin flips are favorable (i.e., will yield a fast execution), it can schedule one more process to invoke $op$, increasing the contention to more than $k$ processes.
This would prevent ``fast'' executions from contributing to $\Exp[\max_{p} T_{op,p}(\EE_{M,A}) \mid K \leq k]$.

We define a measure of contention, called \emph{max-contention}, that the adversary cannot change once the first process is poised to invoke operation $op$.
Let $\EE$ be an execution of algorithm $M$,
%in which each process invokes $op$ at most once.
and let $\EE'$ be the prefix of $\EE$ ending when the first process becomes poised to invoke $op$; $\EE':=\EE$ if no such process exists.
The \emph{max-contention} of $op$ in execution $\EE$ of algorithm $M$, denoted $k_{\max}^{M,op}(\EE)$, is the maximum number of processes that invoke $op$, in \emph{any execution of $M$ that is an extension of $\EE'$}.
In other words, %given the prefix $\EE'$ of $\EE$ that ends as soon as some process becomes poised to invoke $op$, %($\EE'=\EE$ if no such process exists),
$k_{\max}^{M,op}(\EE)$ is the maximum number of invocations of $op$ for any possible way of continuing execution $\EE'$ of $M$.

Let $\mathit{Exec}_{M,A,op}(k)$
%\gtodo{get rid of $A$? in formula below assume expectation is zero if conditioning has zero probability}
be the set of all possible executions $\EE'$ of algorithm $M$ that can result for a given adversary $A$, and have the properties that: (i)~$\EE'$ ends when the first process becomes poised to invoke $op$; and (ii)~$k_{max}^{M,op}(\EE')\leq k$.
We define the \emph{adaptive expected \maxstep{} complexity} of $op$ in algorithm $M$
against adversary~$A$ to be a function $\tau:\{1,\dots,n\}\to\IIR_{\geq 0}$, where
\begin{equation}
    \label{eq:def_max_contention}
    \tau(k)
    :=
    \sup_{\EE' \in \mathit{Exec}_{M,A,op}(k)}
    \Exp\left[\max_{p} T_{op,p}(\EE_{M,A}) \ \middle\vert\
    \text{$\EE_{M,A}$ is an extension of $\EE'$}
%    k_{\max}^{M,op}(\EE_{M,A}) \leq k
    \right]
    .
\end{equation}
The adaptive expected \maxstep{} complexity of $op$ against adversary $A$ (or against an adversary model $\AA$)
is defined similarly to $\tau(k)$, except that the supremum is taken also
%as the supremum of $\tau$
over all algorithms $M$ (respectively, over all $M$ and all $A\in \AA(M)$).
%
%\gtodo{I slightly modified these definitions}
We say that the \emph{adaptive \maxstep{} complexity} of $op$ in algorithm $M$
against adversary~$A$ is bounded by $b(k)$ with probability $q(k)$, if
\[
%    \sup_{\EE' \in \mathit{Exec}_{M,A,op}(k)}
    \Pr\left(\max_{p} T_{op,p}(\EE_{M,A}) \leq b(k) \ \middle\vert\
    \text{$\EE_{M,A}$ is an extension of $\EE'$}
    \right)
    \geq
    q(k)
    ,
    \text{ for all }
        \EE' \in \mathit{Exec}_{M,A,op}(k)
    .
\]
We also say that the adaptive \maxstep{} complexity of $op$ against ~$A$ (or $\AA$) is bounded by $b(k)$ with probability $q(k)$, if the above holds for all algorithms $M$ (respectively, all $M$ and all $A\in \AA(M)$).
Throughout the remainder of the paper, when we say \emph{(expected) \maxstep{} complexity}, we mean \emph{adaptive (expected) \maxstep{} complexity}.
%This choice of wording is justified by the fact that if an implementation has adaptive step complexity $\tau(k)=c$, where $c$ is independent of $k$ (but $c$ may be a function of the number $n$ of processes in the system), then the implementation also has (non-adaptive) expected step complexity $c$.
%Similarly, any upper or lower bound $b(k)$ on the expected \maxstep{} complexity of an implementation, which does not depend on $k$, is also an upper or lower bound on the (non-adaptive) \maxstep{} complexity of that implementation.
%

In the terminology introduced in this section, we will often replace operation $op$ by the object $G$ that supports this operation, if $op$ is the only operation that $G$ provides.

\subsection{Some Basic Objects}
\label{sec:prelim-objects}

We now describe several simple objects that we use as building blocks for our TAS algorithms.

A \emph{doorway} object
supports the operation \enter{} which takes no parameters and returns a boolean value, \True or \False.
Each process calls \enter{} at most once, and we say that it \emph{enters} the doorway when it invokes \enter{}, and \emph{exits} when the \enter{} method responds.
The process \emph{passes through} the doorway if its \enter{} method returns \True, and is \emph{deflected} if it returns \False.
A doorway object satisfies the following two properties:
\begin{enumerate}[label=(D\arabic*)]
  \item \label{D1} Not all processes entering the doorway are deflected; and
  \item \label{D2} If a process passes through the doorway, then it entered the doorway before any process exited the doorway.
\end{enumerate}
A simple, wait-free implementation of a doorway object is given in Figure~\ref{fig:doorway}.
It is straightforward that the implementation satisfies properties \ref{D1} and \ref{D2}: The first process that writes to $B$ ``closes'' the doorway.
All processes that read $B$ after that will be deflected, and thus \ref{D2} is true.
But the first process that reads $B$ does not get deflected, because at the point of that read, no process has written $B$.
Therefore, \ref{D1} is also true.
The implementation uses only one register and each process finishes its \enter{} method in a constant number of steps.
\begin{classfigure}[t!]\small
  \begin{class}{Doorway}
    \shared{register $B\gets\False$}
  \end{class}
  \begin{method}{enter()}
    \If{$B.\Read{}=\False$\label{ln:DW:read}}{
      $B$.\Write{\True}\label{ln:DW:write}\;
      \Return{\True}
    }{
      \Return{\False}
    }
  \end{method}
  \caption{A doorway implementation.\label{fig:doorway}}
\end{classfigure}

\begin{classfigure}[t!]\small
    \begin{class}{Splitter}
  \shared{register $X$; \Doorway $D$}
  \end{class}
  \begin{method}{split()}
    $X$.\Write{\myID}\;
    \If{$D$.\enter{}\label{line:splitter:doorway}}{
      \IlIf{$X.\Read{}=\myID$}{
	\Return{\Stop}\;
      }
      \Return{\Right}
    }
    \Return{\Left}\;
%    \IlIf{$D.\enter{}=\False$}{\Return{\Left}}\;
%    \IlIf{$X.\Read{}=\myID$}{
%      \Return{\Stop}\;
%    }
%    \Return{\Right}\;
  \end{method}\bigskip

  \begin{class}{RSplitter}
  \shared{register $X$; \Doorway $D$}
  \end{class}
  \begin{method}{split()}
    $X$.\Write{\myID}\;

%    \IlIf{$D$.\enter{} {\bf and} $X.\Read{}=\myID$}{
%	  \Return{\Stop}\;
%    }

    \If{$D$.\enter{}}{
      \IlIf{$X.\Read{}=\myID$}{
	    \Return{\Stop}\;
      }
    }

    Choose $dir\in\{\Left,\Right\}$ uniformly at random\;
    \Return{$dir$}\;
  \end{method}

  \caption{Deterministic and randomized splitter implementations.\label{fig:splitter}}
\end{classfigure}

A \emph{randomized 2-process TAS object} can be implemented from a constant number of registers, so that its \TAS{} method has constant expected \maxstep{} complexity.
More precisely, an implementation by Tromp and Vit{\'a}nyi~\cite{Tromp2002}  uses two single-reader single-writer registers, and guarantees for any strong adaptive adversary and any $\ell > 0$, that with probability at least $1-1/2^\ell$, both processes finish after $O(\ell)$ steps.
In our algorithms, when a process calls the \TAS{} method of a 2-process TAS object it must ``simulate'' one of two possible IDs, $1$ or $2$.
Thus, we use a 2-process TAS object \TASII that supports an operation \TAS{$i$}, where $i\in\{1,2\}$.
If two processes call the method \TAS{$i$}, they must use different values for $i$.
We will say that a process \emph{wins} (\emph{loses}) if its \TAS{} call returns 0 (respectively 1).

A \emph{splitter object}~\cite{Moir1994-splitter,Attiya2006-rsplitter} provides a single method \split{}, which takes no parameters and returns a value in $\{\Stop,\Left,\Right\}$.
If a process $p$ calls \split{}, we say that $p$ \emph{goes through the splitter}.
If the call returns \Stop, we say that $p$ \emph{stops} at the splitter; and if it returns \Left (\Right), we say $p$ \emph{turns left} (respectively \emph{right}).

A \emph{deterministic splitter}, denoted \Splitter, was proposed by Moir and Anderson~\cite{Moir1994-splitter}.
It guarantees that if $\ell$ processes go through the splitter, then at most $\ell-1$ turn left, at most $\ell-1$ turn right, and at most one stops.
Thus if only one process goes through the splitter, that process stops.

A \emph{randomized splitter}, denoted \RSplitter, was proposed by Attiya, Kuhn, Plaxton, Wattenhofer and Wattenhofer~\cite{Attiya2006-rsplitter}.
Similarly to the deterministic splitter, it guarantees that if only one process goes through the splitter, then that process must stop.
But now, any process that does not stop, turns left or right with equal probability, and independently of other processes.
Randomized and deterministic splitters are incomparable in ``strength'', as for a randomized splitter it is possible that all processes going through it turn to the same direction.

Both splitter implementations, the deterministic one by Moir and Anderson, and the randomized by Attiya et.~al., use two shared registers and have \maxstep{} complexity $O(1)$ in any execution.
For completeness we provide the implementations in Figure~\ref{fig:splitter}.
%
%They have the following additional doorway-like property, which is useful for the design of  our algorithms:
%\begin{quote}
%  If a \split{} call returns \Stop, then it was invoked before any other \split{} call responded.
%\end{quote}
%\pwtodo{Add an argument}
%
The deterministic splitter implementation has the following additional doorway-like property, which is useful for the design of our algorithms:
\begin{enumerate}[label=(S)]
  \item \label{S} If a process stops or turns right at the splitter, then its  \split{} call was invoked before any other \split{} call on the same object responded.
\end{enumerate}
This follows immediately from the use of doorway $D$ in line~\ref{line:splitter:doorway}:
Suppose a \split{} operation by process $p$ gets invoked after some other \split{} call by process $q$ responded.
Then $q$ has already exited the doorway, when $p$ enters it, so by doorway-property \ref{D2} process $p$ gets deflected, and its \split{} call returns \Left.

\section{Fast TAS for Weak Adversaries}
\label{sec:algos-oblivious}

%\pwtodo{Some parts have to be removed from this section.}

We present implementations of TAS objects for weak adversary models.
In Section~\ref{sec:group-election}, we introduce the problem of group election, which is a natural variant of leader election, and
in Section~\ref{sec:le-from-ge}, we give a TAS implementation from group election objects.
Then, in Sections~\ref{sec:location-algo} and~\ref{sec:rw-algo}, we  provide efficient randomized implementations of group election from registers, for the location-oblivious and the r/w-oblivious adversary models, respectively.
%\todo{Perhaps give an overview at the end of the introduction, instead.}

\subsection{Group Election}
\label{sec:group-election}
%To implement our TAS algorithms, we define the \emph{group election problem}.
%The idea is that group election reduces the number of participating processes, but also guarantees that at least one process remains.
In the \emph{group election problem} processes must elect a non-empty subset of themselves, but
%In the \emph{group election} problem, processes
unlike in leader election, it is not required that exactly one process gets elected.
Still it is desirable that the \emph{expected} number of processes elected should be bounded by a small function in the number of participating processes.

Formally, a group election object, denoted \GroupElect, provides the method \elect{}, which takes no parameters and returns either \win or \lose.
We say a process \emph{participates} in a group election when it calls \elect{}.
The processes whose \elect{} calls return \win get \emph{elected}.
A group election object must satisfy the following property:
\begin{enumerate}[label=(GR)]
  \item\label{GR}
  Not all participating processes' \elect{} calls return \lose.
\end{enumerate}
That is, if at least one process participates and all participating processes finish their \elect{} calls, then at least one process gets elected.

%\gtodo{the beginning of this para is a bit repetitive; I suggest we replace this para with the one I added below}
%The aim of a group election object is to force as many processes as possible to lose, i.e., as few processes as possible should get elected in expectation.
%We are interested in group election objects, where the number of surviving processes is a function of the max-contention.
%This function is called \emph{effectiveness} and is defined as follows.
%
%
%We require that in each \elect{} method call,
%\gtodo{we no longer need the doorway inside elect()}
%the calling process first enters a doorway and it loses the group election if it is deflected.
%This way, once the first process exits the doorway the set of processes that may pass through the doorway and continue with the group election is bounded by the set of processes that have already entered the doorway.
%This prevents the adversary from increasing contention on the remainder of the group election at a later point.
%We need this property to obtain group election algorithms that are adaptive (w.r.t.\ their efficiency) to contention.
%We will also assume that in any execution, each process finishes its \enter{} method within a constant number of its own steps, which can be achieved by using the doorway implementation given in Figure~\ref{fig:doorway}.

We are interested in group election objects for which the expected number of elected processes is bounded by a (small) function of the max-contention.
This function is called the \emph{effectiveness} of the group election object and is formally defined next.

Consider an $n$-process
%implementation $G$ of a group election object.
group election object $G$.
Let $M$ be an algorithm in which processes invoke the \elect{} operation of $G$,
and let $A$ be an adversary.
For an execution $\EE$ of $M$, let $\mathit{win}(\EE)$ denote the number of processes that get elected on $G$.
Similarly to definition~\eqref{eq:def_max_contention} for \maxstep{} complexity, let $\mathit{Exec}_{M,A,G}(k)$ be the set of all possible executions $\EE'$ of $M$ that can result for adversary $A$, and have the properties that: (i)~$\EE'$ ends when the first process is poised to invoke $G$.\elect{}; and (ii)~$k_{max}^{M,G}(\EE')\leq k$.
%Recall that $k_{max}^{M,G}(\EE')$ is the max-contention of $G$ on $\EE'$, and $\EE_{M,A}$ is a random execution of $M$ scheduled by $A$.
The \emph{effectiveness} of group election object $G$ in algorithm $M$ against adversary $A$ is a function $\varphi:\{1,\dots,n\}\to [1,n]$, where
\begin{displaymath}
\varphi(k):=
\sup_{\EE'\in\mathit{Exec}_{M,A,G}(k)}
%\left\{
  \Exp\left[
    \mathit{win}(\EE_{M,A})
    \ \middle\vert\
    \text{$\EE_{M,A}$ is an extension of $\EE'$}
%    k_{max}^{M,G}(\EE_{M,A})\leq k
  \right]
%\right\}
.
\end{displaymath}
%Note that effectiveness is defined in the same way as the expected max-contention in (\ref{eq:def_max_contention}), but for the random variable $win(\EE_{M,A})$ instead of $\max_{p} T_{G.\elect{},p}(\EE_{M,A})$.
%Note the similarity of this definition to the definition of the expected max-contention in~\eqref{eq:def_max_contention}.
The effectiveness of  $G$ against adversary $A$ (or against an adversary model $\AA$)
is defined similarly to $\phi(k)$, except that the supremum is taken also
%is the supremum of $\phi(k)$
over all algorithms $M$ (respectively, over all $M$ and $A\in \AA(M)$).

\subsection{TAS from Group Election}
\label{sec:le-from-ge}

\begin{classfigure}[t]\small
  \begin{class}{TAS}
  \shared
  {\GroupElect $G[1\dots n]$;
    \Splitter $S[1\dots n]$; %\hspace{3in}
    \TASII $T[1\dots n]$;
    \Doorway $D$
%    ;
%    \Doorway $D[1\ldots n]$
  }
  \end{class}
%\gtodoi{Add a doorway in line 0; this is not necessary as the first step a process executes in the TAS() method below is to enter the doorway of $G[1]$}
  \begin{method}{TAS()}
    \IlIf{$D.\enter{} = \False$}{\Return{$1$}}\label{ln:LE:doorway}\;
    $i \gets 0$\;
    \Repeat{$s=\Stop$\label{ln:LE:until}}{\label{ln:LE:repeat}
        $i\gets i+1$\;
        \IlIf{%$D[i].\enter{} = \False$ {\bf or}
        $G[i].\elect{}=\lose$}{\Return{$1$}}\label{ln:LE:elect}\;
        $s\gets S[i].\split{}$\label{ln:LE:split}\;
        \IlIf{$s=\Left$}{\Return{$1$}}\;
    }
    \IlIf{$T[i].\TAS{$1$}=1$}{\Return{$1$}}\label{ln:LE:LE2_1}\;
    \While{$i>1$}{
      $i\gets i-1$\;
      \IlIf{$T[i].\TAS{$2$}=1$}{\Return{$1$}}\label{ln:LE:LE2_2}\;
    }
    \Return{$0$}\label{ln:LE:win}\;
  \end{method}
  \caption{An implementation of TAS from group election objects.} \label{fig:generalLE}
\end{classfigure}

We now present an implementation of a TAS object from $n$ group election objects.
The algorithm uses also $n$ deterministic splitters, $n$ 2-process TAS objects, and one doorway object.
All these objects can be implemented from a total number of $O(n)$ registers,
as we saw in Section~\ref{sec:prelim-objects}.

The implementation is given in Figure~\ref{fig:generalLE}.
First, each process enters a doorway, and if deflected, its \TAS{} call immediately returns 1.
Any process that passes through the doorway participates in a series of group elections, on  objects $G[1],\ldots,G[n]$.
If the process is not elected on $G[i]$, then its \TAS{} returns 1.
Otherwise, it goes through splitter $S[i]$ next.
If the process turns left at the splitter, then \TAS{} returns 1; if it turns right, it participates in the next group election, on $G[i+1]$.
Finally, if the process stops at $S[i]$, then it does not participate in any further group elections.
Instead, it tries to win a series of 2-process TAS, on $T[i],\dots,T[1]$, until it either loses in one of them and returns 1, or wins in all of them and returns 0.
%\pwtodo{Should we explain that a process is using either ID 1 or 2, when it executes a 2-process TAS?}
%\pwtodo{I think it should be winning a ``TAS object'', not winning a ``TAS''.}

The idea is that fewer and fewer processes participate in each group election, as only processes that get elected in $G[i]$ may participate in $G[i+1]$.
The rate at which the number of processes drops depends on the effectiveness of the group election objects.
The purpose of the doorway at the beginning is to achieve linearizability (without the doorway, we would obtain a leader election object instead).
%The doorway objects ensure correctness (linearizability).
%and also simplify the analysis of the step complexity.
%In fact, one doorway object, $D[1]$, would be sufficient, but having multiple ones simplifies the analysis of the step complexity.
The splitter objects serve two purposes.
First, they ensure that as soon as only one process remains, that process will
%stop participating in group elections, and will switch to the list of 2-process TAS objects.
not participate in other group elections, and will switch to the list of 2-process TAS objects.
Second, they guarantee that the number of processes participating in each $G[i]$ \emph{strictly} decreases with $i$.
This ensures that no more than $n$ group election objects are needed.
Finally, the 2-process TAS objects ensure that (at most) one process returns $0$.

Next we prove the correctness of the implementation, and analyze its \maxstep{} complexity in terms of the  \maxstep{} step complexity and effectiveness of the group election objects used.

We use the following standard notation.
%\gtodo{move the definitions to preliminaries section?}
For any function $f\colon X \to Y$, where $Y\subseteq X$, and for $i\geq 0$, we denote by $f^{(i)}$ the $i$-fold composition of $f$, defined recursively by $f^{(0)}(x) = x$ and $f^{(i+1)}(x) = f(f^{(i)}(x))$.
Further, if $f$ is a real function, we define
\[
    f^\ast(x) = \inf\{i\colon f^{(i)}(x) \leq 1\}.
\]

\begin{theorem}
    \label{thm:groupelect->leaderelect}\samepage
    Figure~\ref{fig:generalLE} gives an implementation of a TAS object from a set of group election objects.
    Suppose that for each group election object $G[i]$ used in this implementation, the expected \maxstep{} complexity of $G[i]$ against a given adversary $A$ is bounded by a function $t(k)$ of the max-contention $k$ of $G[i]$,
    and the effectiveness of $G[i]$ against $A$ is bounded by a function $f(k)$.
    Suppose also that functions $f$ and $t$ are non-decreasing,
    and $f$ is concave.
    Then the expected \maxstep{} complexity of the TAS implementation against $A$ is
    $
        O\big(t(k)\cdot g^\ast(k) \big),
    $
    where $g(k) := \min\{f(k), k - 1\}$.
%    $f^\ast(x) = \min \{i \colon f^{(i)}(x)\leq x_f \}$
%    and
%    $x_f = \sup\{x\geq 0 \colon f(x) \geq x-1\}$.
    Moreover, the same bound on the expected \maxstep{} complexity applies even if the assumption that the effectiveness of $G[i]$ is bounded by $f(k)$ holds only for $1\leq i \leq  g^\ast(n)$.
\end{theorem}

The assumption that functions $t$ and $f$ are non-decreasing is not restrictive, as the expected \maxstep{} complexity and the effectiveness are by definition non-decreasing functions of the max-contention.
The requirement that $f$ is concave is also reasonable, as it suggests that
%suggests a natural ``diminishing returns'' effect:
the larger the max-contention, the smaller the increase in the expected number of elected processes, for the same increase in max-contention.
This assumption is needed in the analysis for the following reason:
We inductively obtain upper bounds for the expectation $E[k_j]$ of the number $k_j$ of processes participating in the group election on $G[j]$.
Concavity of $f$ allows us to bound the expected number of processes getting elected on $G[j]$ by Jensen's inequality: $E[f(k_j)]\leq f(E[k_j])$.

\subsubsection{Proof of Theorem~\ref{thm:groupelect->leaderelect}}

We first show that the implementation is correct, and then analyze its \maxstep{} complexity.

\begin{proof}[\bf Correctness]
Consider an arbitrary execution.
For $j\geq 1$, let $m_j$ be the number of processes that begin $j$ iterations of the repeat-until loop (lines~\ref{ln:LE:repeat}--\ref{ln:LE:until}), and for $1\leq j\leq n$ let $e_j$ be the number of processes that get elected on group election object $G[j]$ (in line~\ref{ln:LE:elect}).
Clearly, $e_j\leq m_j$, and at most $e_j$ processes go through splitter $S[j]$ (in line~\ref{ln:LE:split}).
Moreover, by the splitter semantics, at most $e_j-1$ of them turn right, provided $e_j\geq 1$.
Thus, if $e_j\geq 1$,  at most $e_j-1\leq m_j-1$ processes execute a $(j+1)$-th iteration of the repeat-until loop.
Hence, $m_{j+1}<m_j$, and in particular $m_{n+1}=0$.
Thus, we have
\begin{equation}\label{eq:TAS_correctnes_i*}
  j^\ast:=\max_j\{m_j\geq 1\}\leq n.
\end{equation}

Next we observe that each 2-process TAS object $T[j]$, $1\leq j\leq n$, is accessed by at most two processes: possibly a single process which stops at $S[j]$ and then calls $T[j]$.\TAS{$1$} (in line~\ref{ln:LE:LE2_1}), and possibly the winner of $T[j+1]$, if $j<n$, which calls $T[j]$.\TAS{$2$} (in line~\ref{ln:LE:LE2_2}).
At most one of them can win $T[j]$.
Since a process needs to win $T[1]$ (either in line~\ref{ln:LE:LE2_1} or in line~\ref{ln:LE:LE2_2}) in order to win the implemented \TAS{} method, it follows that at most one process wins.

We now argue that \emph{at least} one process wins, provided that at least one process calls the implemented \TAS{} method, and that all processes that do so finish their call.
Recall that by \eqref{eq:TAS_correctnes_i*}, $j^\ast\leq n$ is the largest index such that at least one process starts its $j^\ast$-th iteration of the while-loop.
By \ref{GR}, at least one process gets elected on $G[j^\ast]$, and subsequently goes through splitter $S[j^\ast]$.
I.e., $e_{j^\ast}\geq 1$.
Since $m_{j^\ast+1}=0$, none of these $e_{j^\ast}$ processes executes another iteration of the repeat-until loop, and thus they all turn left or stop at $S[j^\ast]$.
As not all of them can turn left either, at least one (and thus by the splitter semantics exactly one) process must stop at $S[j^\ast]$.
It follows that at least one process calls $T[j^\ast]$.\TAS{$1$}, and so at least one process wins $T[j^\ast]$.
Since for $1 < j \leq n$ the winner of $T[j]$ continues to $T[j-1]$, some process must win $T[1]$.
Thus some process wins the implemented \TAS{} method.

It remains to show that the TAS implementation is linearizable.
If process $p$'s \TAS{} call $z$ returns 0, then by property~\ref{D2} of the doorway, $p$ must have entered doorway $D$ during $z$ before any other process exited it.
In particular, no \TAS{} call \emph{happens before} $z$ (i.e., responds before $z$ gets invoked).
Thus, we can obtain a linearization of the execution by putting $z$ first, and adding all other \TAS{} operations after $z$ in the order of their invocation.
The resulting sequential history is valid (the first \TAS{} returns 0, and all other \TAS{} return 1), and preserves the happens-before order, because no \TAS{} happens before $z$.\noqed
\end{proof}

\begin{proof}[\bf Step Complexity]
Consider an algorithm $M$
that uses the implemented TAS object.
%in which processes call the implemented \TAS{} operation.
Let $\EE := \EE_{M,A}$ be a random execution of $M$ scheduled by adversary $A$.

For each $1\leq j\leq n$, let $\EE_j$ be the prefix of execution $\EE$, until the first process is poised to invoke $G[j].\elect{}$; or $\EE_j := \EE$ if no such process exists.
%Let $s_i$ be the last step in $\EE_i$.
Observe that if $\EE_1 \neq \EE$, then the last step of $\EE_1$ is the step at which the first process passes through doorway $D$.
Similarly for $i>1$, if $\EE_i \neq \EE$ then the last step of $\EE_i$ is the step in which the first process turns right at splitter $S[j-1]$.

For $1\leq j \leq n$, let $k_j := k_{max}^{M,G[j]}(\EE)$ be the max-contention of $G[j]$ in $\EE$.
By definition, it is also
\[
    k_j = k_{max}^{M,G[j]}(\EE_j).
\]

Let $\EE_0$ be the prefix of $\EE$ until the first process is poised to invoke the implemented \TAS{} operation, i.e., it is poised to enter doorway $D$.
Let
$k_0:= k_{max}^{M,D}(\EE)$ be the max-contention of $D$ in $\EE$, and thus  $k_0= k_{max}^{M,D}(\EE_0)$ as well.

Observe that, for any $1\leq j\leq n$, execution $\EE_{j-1}$ is a prefix of $\EE_{j}$, and $k_{j-1}\geq k_j$.

Let $T(\EE)$ denote the \maxstep{} complexity of the implemented TAS in execution $\EE$.
To prove the expected \maxstep{} complexity bound claimed in the theorem we must show that for any given $k\geq 0$, if $k_0=k$ then
\[
    \Exp[T(\EE) \mid \EE_0] = O\big(t(k) \cdot g^\ast(k)\big).
\]
We will assume $k_0\geq1$, otherwise $T(\EE)=0$ as no process invokes the implemented \TAS{}.

First we bound the expected number of group election objects accessed by at least one process in the execution.

For $1\leq j\leq n$, let $e_j$ be the number of processes elected in the group election on $G[j]$.
From the theorem's assumption that the effectiveness of $G[j]$ is bounded by function $f$ of the max-contention of $G[j]$, it follows %that for $1\leq j\leq n$,
\[
    \Exp[e_j \mid \EE_{j}] \leq f(k_{j}).
\]
We take the conditional expectation given $\EE_0$ to obtain
\[
    \Exp[\Exp[e_j \mid \EE_j]\mid \EE_0] \leq \Exp[f(k_j)\mid \EE_0].
\]
The expression on the left equals $\Exp[e_j \mid \EE_0]$ by the tower rule, since $\EE_0$ is a prefix of $\EE_j$.
For the right side  we have $\Exp[f(k_{j})\mid \EE_0]\leq f(\Exp[k_{j}\mid \EE_0])$, by Jensen's inequality and the assumption that $f$ is concave.
Therefore,
\begin{equation}
    \label{eq:ExpejEE0}
    \Exp[e_{j}\mid \EE_0] \leq f(\Exp[k_{j}\mid \EE_0]).
\end{equation}

For $j=1$, \eqref{eq:ExpejEE0} yields
\[
    \Exp[e_{1}\mid \EE_0] \leq f(\Exp[k_{1}\mid \EE_0])\leq f(\Exp[k_{0}\mid \EE_0]) = f(k_0),
\]
where the second inequality holds because $k_1\leq k_0$ and $f$ is non-decreasing, and the last equation holds because $k_0$ is completely determined given $\EE_0$.

For $j>1$, we have $k_j\leq e_{j-1}$: This is trivial if $k_j=0$.
If $k_j\geq 1$ then the last step of $\EE_{j}$ is when the first process $p$ turns right at splitter $S[j-1]$.
Property~\ref{S} then implies that any other process $q$ that may participate at the group election in $G[j]$ must have already invoked $S[j-1]$.\split{}, and thus must have already been elected at $G[j-1]$.

Using the inequality $k_j\leq e_{j-1}$ we have just shown, and the assumption $f$ is non-decreasing, we obtain from~\eqref{eq:ExpejEE0} that for $j>1$,
\[
    \Exp[e_j \mid \EE_0] \leq f(\Exp[e_{j-1}\mid \EE_0]).
\]

Combining the above inequalities for $j=1,2,\ldots,$ and using that $f$ is non-decreasing we get
\[
    \Exp[e_{j} \mid \EE_0] \leq f^{(j)}(k_0).
\]
The $e_{j}$ processes elected on $G[j]$ will participate in an additional number of at most $e_{j}$ group election objects beyond the first $j$ ones, as each splitter $S[i]$ ensures $e_{i+1}\leq e_{i}-1$, if $e_{i} > 0$.
Therefore, if $j^\ast := \max\{j\colon k_j >0\}$ is the total number of group election objects accessed by at least one process, then for any $0 \leq j\leq k_0$,
\begin{equation}
    \label{eq:ExpejEE0x}
    \Exp[j^\ast \mid \EE_0] \leq j + f^{(j)}(k_0).
\end{equation}
Let
\[
    x := \max\{y \leq k_0 \colon f(y) \geq y-1\},\quad
    \lambda := \min \{i \colon f^{(i)}(k_0)\leq x \}
.
\]
Note that $x\geq 1$, as $f(y)\geq 0$ for $y\geq 0$.
Also, since $f$ is concave and non-negative, it follows
\[
    f(y) \geq y-1, \text{ for } 0 \leq y \leq x.
\]
Setting $j := \lambda$ in~\eqref{eq:ExpejEE0x} we obtain
$\Exp[j^\ast \mid \EE_0] \leq \lambda + f^{(\lambda)}(k_0)$.
Since $f(y) \geq y-1$ for $0 \leq y \leq x$, and by definition, $f(y) < y-1$ for $x < y \leq k_0$,
it follows that $\lambda + f^{(\lambda)}(k_0) \leq g^\ast(k_0)  + 1$, where $g(k) := \min\{f(k), k - 1\}$.
Therefore,
\begin{equation}
    \label{eq:expiast}
    \Exp[j^\ast \mid \EE_0]
    \leq
    g^\ast(k_0)+1.
\end{equation}

In the following we will assume that $\EE_0$ is fixed, thus so is $k_0$.

Next we will bound the expectation of the maximum number of steps any single process takes on the group election objects.
This number is bounded by $\sum_{1\leq j\leq j^\ast} t_j$, where $t_j$ is the \maxstep{} complexity of $G[j]$ in $\EE$.
We will bound the expectation of this sum using a version of Wald's Theorem (note that the number $j^\ast$ of terms in the sum as well as the terms $t_j$ are random variables.)
From the assumption that the \maxstep{} complexity of $G[j]$ is bounded by a function $t$ of the max-contention on $G[j]$, we have that
\[
    \Exp[t_j \mid \EE_j] \leq t(k_j).
\]
Since $k_j\leq k_0$ and $t$ is a non-decreasing function, it follows $\Exp[t_j \mid \EE_j] \leq t(k_0)$.
This implies %that if $\Pr(j^\ast \geq j) > 0$,
%\pwtodo{We are using bold for $\Exp$, but regular font for Pr(). We should make this look the same. Also, I started using \Prob() later on, so we need to change one of them}
\[
    \Exp[t_j \mid j^\ast \geq j] \leq t(k_0),
\]
as the execution prefix $\EE_j$ is sufficient to determine whether or not $j^\ast \geq j$ holds.
We will use the above inequality to apply the
%We will now use the
following variant of Wald's Theorem, for random variables that are not independent.
A proof of this theorem can be found, e.g., in~\cite{Jae2007a}.

\begin{theorem}[Wald's Theorem]
    Let $X_1,X_2,\dots$ be a sequence of non-negative random variables and let $Y$ be a non-negative integer random variable such that the expectations of $Y$ and of each $X_j$ exist.
    If for all $j$, $\Exp[X_j \mid j\leq Y] \leq \mu$ for some $\mu \geq 0$, then
    $\Exp[X_1+\dots+X_Y] \leq \mu\cdot \Exp[Y]$.
\end{theorem}
We apply the theorem for $X_j=t_j$, $Y = j^\ast$, and $\mu = t(k_0)$ to obtain
\[
    \Exp\left[\sum_{1\leq j\leq j^\ast}t_j\right]
    \leq
    t(k_0)\cdot \Exp[j^\ast]
%    \stackrel{\eqref{eq:expiast}}
%    \leq
%    t(k_0)\cdot (f^\ast(k_0) + x_f)
    .
\]

Using the same argument we can also bound the expectation of $\sum_{1\leq j\leq j^\ast} t'_j$, where $t'_j$ is the \maxstep{} complexity of the \TASII object $T[j]$ in $\EE$.
For $T[j]$ we have that its expected \maxstep{} complexity is constant (against any adversary), i.e., $\Exp[t'_j \mid \EE'_j] = O(1)$, for the prefix $\EE'_j$ of $\EE$ until some process is poised to invoke $T[j]$.\TAS{}.
Then the same reasoning as above yields
\[
    \Exp\left[\sum_{1\leq j\leq j^\ast}t'_j\right]
    =
    O(1)\cdot \Exp[j^\ast]
    .
\]

Finally, the number of remaining steps of a process in $\EE$, that are not steps on one of the objects $G[j]$ or $T[j]$, is bounded by $O(j^\ast)$.

Therefore the expected \maxstep{} complexity of the TAS implementation is bounded by
\[
    t(k_0)\cdot \Exp[j^\ast]
    +
    O(1)\cdot \Exp[j^\ast]
    +
    O(\Exp[j^\ast])
    \stackrel{\eqref{eq:expiast}}
    =
    O\big(t(k_0)\cdot g^\ast(k_0)\big).
\]

Finally, note that for the above analysis we do not need any assumptions on the effectiveness of objects $G[j]$ for $j > g^\ast(n)$, as \eqref{eq:ExpejEE0x} is used only for $j := \lambda \leq g^\ast(k_0)$.
This completes the proof of Theorem~\ref{thm:groupelect->leaderelect}.
%\noqed
\end{proof}

\subsection{Group Election for Location-Oblivious Adversaries}
\label{sec:location-algo}

\begin{classfigure}[t]\small
  \begin{class}{GroupElect}
  \tcc{$\ell:=\lceil\log n\rceil$}
  \shared{register $R[1\dots\ell+1]\gets [0\dots0]$}
%     \qquad \Doorway $D$\;
%        \label{ln:location-oblivious:read-flag}\;
%    \qquad \integer $flag\gets0$
%        \label{ln:location-oblivious:write-flag}\;
  \end{class}
  \begin{method}{elect()}
%    \IlIf{$flag.\Read{}=1$}{\Return{\False}}\;
%    $flag$.\Write{$1$}\;
%     \IlIf{$D$.\enter{}=\False}{\Return{\False}}\label{ln:location-oblivious:doorway}\;
    Choose $x\in \{1,\dots,\ell\}$ at random
    such that
    $\Pr(x = i) = 2^{-i}$ for $1\leq i<\ell$, and $\Pr(x = \ell) = 2^{-\ell+1}$
    \label{ln:location-oblivious:random}\;
    $R[x]$.\Write{$1$}\label{ln:location-oblivious:write}\;
    \IlIf{$R[x+1].\Read{}=0$}{\Return{\win}}
        \label{ln:location-oblivious:read}\;
    \Return{\lose}\;
  \end{method}
  \caption{A group election implementation for the location-oblivious adversary model.}
  \label{fig:GroupElect-location-oblivious}
\end{classfigure}

We present a simple randomized group election implementation from registers, which has effectiveness $O(\log k)$ in the location-oblivious adversary model, and constant \maxstep{} complexity.
This can be used to implement a TAS object with expected \maxstep{} complexity $O(\log^\ast k)$ against location-oblivious adversaries.

The group election implementation is given in Figure~\ref{fig:GroupElect-location-oblivious}.
Each process first writes
to a random register among the $\ell:=\lceil\log n\rceil$ registers $R[1],\ldots,R[\ell]$, where $R[i]$ is chosen with probability $1/2^i$ if $1\leq i<\ell$, and with probability $1/2^{\ell-1}$ if $i=\ell$.
Then the process reads the next register, $R[i+1]$, and gets elected if and only if no process has previously written to that register.

%Consider a random execution in which processes access the implemented group election objected, scheduled by a location-oblivious adversary.
We have that at least one process gets elected, namely a process that writes to the rightmost register %$R[i^\ast]$, for the largest index $i^\ast$ such that $R[i^\ast]$
that gets written.
The idea for the $O(\log k)$ bound on the effectiveness is as follows.
Since the probability that a process chooses index $i+1$ equals half the probability it chooses $i$,
at most a constant expected number of processes write to $R[i]$ before some process writes to $R[i+1]$.
% once a constant expected number of processes have written to $R[i]$, some process has also written to $R[i+1]$.
After a process has written to $R[i+1]$, no process that writes to $R[i]$ can still get elected.
Therefore, for every index $i$ there will only be a constant expected number of processes that choose that index and get elected.
Moreover, if at most $k$ processes participate in the group election, then with sufficiently high probability (in $k$) only the first $O(\log k)$ registers get written at all.
A simple calculation then shows that only an expected number of $O(\log k)$ processes get elected.

\begin{lemma}
    \label{lem:GroupElect-location-oblivious}
    Figure~\ref{fig:GroupElect-location-oblivious} gives a randomized implementation of a group election object %from $O(\log n)$ registers,
    with effectiveness at most $2\log k + 4$ and constant \maxstep{} complexity against any location-oblivious adversary.
\end{lemma}

\begin{proof}
Let $M$ be an algorithm that uses the implemented group election object, and consider any execution of $M$ in which all processes participating in the group election finish their \elect{} call.
Let $i^\ast$ be the largest index such that some process $p$ writes to $R[i^\ast]$ (in line~\ref{ln:location-oblivious:write}).
Then $p$ reads the value 0 from $R[i^\ast + 1]$ in the next line and returns \win.
Hence, at least one process gets elected.
Further, %the implementation uses just $\ceil{\log n} + 1$ shared registers, and
each process does exactly two shared memory operations, thus the \maxstep{} complexity is constant.
It remains to bound the effectiveness of the group election object.

Let $A$ be a location-oblivious adversary, and let $\EE := \EE_{M,A}$ be a random execution of $M$ scheduled by $A$.
Fix the prefix $\EE'$ of $\EE$ until the first process is poised to invoke \elect{}, and let $k := k_{max}^{M,\elect{}}(\EE')=k_{max}^{M,\elect{}}(\EE)$ be the max-contention of \elect{} in $\EE$.
Let $k'\leq k$ be the actual number of processes that execute the write operation in line~\ref{ln:location-oblivious:write} during $\EE$, and
for $1\leq i\leq k'$, let $p_i$ be the $i$-th process to execute the write operation. %in line~\ref{ln:location-oblivious:write}.
%and  let $x_i$ be the index of the register on which $p_i$ writes.

Since adversary $A$ is location-oblivious, it does not know the index of the register on which $p_i$ will write, before $p_i$ finishes that operation.
We can thus assume that a list $x_1,\ldots,x_k$ of indices is chosen in advance, right after the last step of $\EE'$, such that each index $x_i$ is drawn independently at random according to the distribution in line~\ref{ln:location-oblivious:random}, and then for each $1\leq i\leq k'$, process $p_i$ writes to register $R[x_i]$ in line~\ref{ln:location-oblivious:write}.
%The list is not revealed to the adversary, who learns $x_i$ only after $p_i$ has written to $R[x_i]$.
Note that although just the first $k'$ of the values $x_i$ are actually used, we draw $k\geq k'$ values initially as $k'$ may not be known in advance.

For each $1\leq i\leq k'$, let $X_i$ be the 0/1 random variable that is $1$ if and only if $p_i$ gets elected, i.e.,  $p_i$ reads the value 0 on register $R[x_i+1]$ in line~\ref{ln:location-oblivious:read}, and returns \win.
Further, for each $1\leq i\leq k$, let $Y_i$ be the 0/1 random variable that is $1$ if and only if $x_{j} \neq x_i+1$ for all $j < i$.
Clearly, $X_i\leq Y_i$ for $i\leq k'$, as $p_i$ reads the value 0 only if none of the processes $p_1,\ldots,p_{i-1}$ writes to register $R[x_i+1]$.
%i.e., $x_{j} \neq x_i+1$ for all $j < i$.
%\pwtodo{For $i\leq k'$, isn't $X_i=Y_i$?}
The expected number of processes that get elected is then
\begin{equation}
    \Exp\Biggl[\sum_{1\leq i\leq k'} X_i\Biggr]
    \leq
    \Exp\Biggl[\sum_{1\leq i\leq k} Y_i\Biggr]
    =
    \sum_{1\leq i\leq k} \Exp[Y_i].\label{eq:oblivious:10}
\end{equation}
Using that $x_1,\dots,x_i$ are chosen independently ($\ast$), and that $x_i=\ell$ implies $Y_i = 1$ ($\dagger$), we obtain
\begin{align*}
    \Exp[Y_i]
    &=
    \Pr\Biggl(\bigwedge_{1\leq j < i}  \left(x_{j} \neq x_i+1\right) \Biggr)
    \\ &=
    \sum_{1\leq x\leq \ell} \Pr\Biggl(\left(x_i = x\right) \,\wedge\, \bigwedge_{1\leq j < i}  \left(x_{j} \neq x+1\right)\Biggr)
    \\ & \stackrel{\text{($\dagger$)}}{=}
    \sum_{1\leq x<\ell} \Pr\Biggl(\left(x_i = x\right) \,\wedge\, \bigwedge_{1\leq j < i}  \left(x_{j} \neq x+1\right)\Biggr)
    +
    \Pr(x_i=\ell)
    \\& \stackrel{\text{($\ast$)}}{=}
    \sum_{1\leq x< \ell} \Pr(x_i = x)\prod_{j= 1}^{i-1} \Pr(x_{j} \neq x+1) +\frac1{2^{\ell-1}}
    \\& =
    \sum_{1\leq x <\ell} \frac1{2^x} \left(1-\frac1{2^{x+1}}\right)^{i-1}
    +\frac1{2^{\ell-1}}.
\end{align*}
Substituting that to (\ref{eq:oblivious:10}) yields
\begin{align*}
    \Exp\Biggl[\sum_{1\leq i\leq k'} X_i\Biggr]
    &\leq
    \sum_{1\leq j\leq k}\Biggl(
    \sum_{1\leq i<\ell} \frac1{2^i} \left(1-\frac1{2^{i+1}}\right)^{j-1}
    +\frac1{2^{\ell-1}}
    \Biggr)
    \\& =
    \sum_{1\leq i<\ell}
    \frac1{2^i}
    \sum_{1\leq j\leq k} \left(1-\frac1{2^{i+1}}\right)^{j-1}
    +
    \sum_{1\leq j\leq k}\frac 1{2^{\ell-1}}
    \notag
    \\& =
    \sum_{1\leq i<\ell}
    \frac1{2^i} \cdot
        \frac{1-\left(1-\frac1{2^{i+1}}\right)^{k}}{1/2^{i+1}}
    +
    \frac k{2^{\ell-1}}
    \notag
    \\&
    = 2\sum_{1\leq i<\ell}
    \left(1-\left(1-\frac1{2^{i+1}}\right)^{k}\right)
    +
    \frac k{2^{\ell-1}}.
\end{align*}
We bound the sum in the last line by bounding with 1 each of the first $\log k$ terms, and using for the remaining terms that $1-\left(1-\frac1{2^{i+1}}\right)^{k} \leq 1-\left(1-\frac k{2^{i+1}}\right) = \frac k{2^{i+1}}$.
%which follows from the known inequality $(1-\epsilon)^x \geq 1-\epsilon x$.
We get
\begin{align*}
    \Exp\Biggl[\sum_{1\leq i\leq k'} X_i\Biggr]
    &\leq
    2\sum_{1\leq i<\log k} 1
    +
    2\sum_{\log k \leq i<\ell} \frac k{2^{i+1}}
    +
    \frac k{2^{\ell-1}}
    \\
    &\leq
    2\log k
    +
    2 \frac k{2^{\log k}}
    +
    \frac k{2^{\ell-1}}
    \\&
    \leq
    2\log k + 4,
  \end{align*}
  as $\ell = \lceil\log n\rceil\geq \log k$.
  This completes the proof of Lemma~\ref{lem:GroupElect-location-oblivious}.
\end{proof}

We can now apply Theorem~\ref{thm:groupelect->leaderelect} to obtain the following result.

\begin{theorem}
    \label{thm:location-oblivious}
    There is a randomized implementation of a TAS object from $\Theta(n)$ registers with expected \maxstep{} complexity $O(\log^\ast k)$ against any location-oblivious adversary.
\end{theorem}
\begin{proof}
We consider the TAS implementation of Figure~\ref{fig:generalLE}, and use the algorithm in Figure~\ref{fig:GroupElect-location-oblivious} to implement the group election objets $G[j]$, for $1\leq j \leq 2\log^\ast n$.
For $2\log^\ast n < j \leq n $, we just let $G[j]$ be a trivial group election object, where all participating processes get elected and the \maxstep{} complexity is zero.
From Lemma~\ref{lem:GroupElect-location-oblivious}, the group election objects $G[j]$, for $1\leq j\leq 2\log^\ast n$, have constant \maxstep{} complexity, and effectiveness  bounded by $f(k) = 2\log k  + 4$ against any location-oblivious adversary.
For $g(k) := \min\{2\log k  + 4, k-1\}$, we have $g^\ast(k) = \log^\ast k + O(1) < 2\log^\ast n$.
Theorem~\ref{thm:groupelect->leaderelect} then implies that the resulting TAS object has expected \maxstep{} complexity $O(\log^\ast k)$ against any location-oblivious adversary.
Moreover the algorithm uses $\Theta(n)$ registers, as each of the first $2\log^\ast n$ group election objection requires $\log n + O(1)$ registers, while the remaining trivial group election objects do not use any registers.
\end{proof}

%Using a standard argument
%With a slightly more careful argument
%one can show that w.h.p.\ only the first $O(\log n)$ group election objects are used.
%This fact has recently been exploited to combine our leader election algorithm with one that uses only $O(\sqrt n)$ registers but has higher expected \maxstep{} complexity, to achieve the same complexity as ours \cite{GHHW2013a}.

\subsection{Group Election for R/W-Oblivious Adversaries}
\label{sec:rw-algo}

\begin{classfigure}[t]\small
  \begin{class}{GroupElect}
  \tcc{$b:=\frac32$ and $\ell:=\lceil\log_{b}\log n\rceil$}
%  and $q_i = 1/2^{b^{i-1}}$, for $1\leq i \leq \ell$  }
  \shared
%     \qquad \Doorway $D$\;
    {register $\Up[1\dots \ell]\gets [0\dots 0]$, $\Down[1\dots \ell-1]\gets [0\dots 0]$}
  \end{class}
  \begin{method}{elect()}
%     \IlIf{$D$.\enter{}=\False}{\Return{\False}}\label{ln:RW-oblivious:doorway}\;
    $i\gets 0$\;
    \Repeat{$c_i = \tails$ {\bf or} $i = \ell$}{
        $i\gets i+1$\;
        Choose $c_i\in\{\heads,\tails\}$ at random such that $\Pr(c_i = \heads) = q_i :=
         1/2^{b^{i-1}}$\;
%        Flip a biased coin $c$ that comes up \heads with probability $q_i := 1/2^{b^{i-1}}$\;
            \uIf{$c_i = \heads$}
                {$\Up[i]$.\Write{$1$}\label{ln:re-oblivious:writeUp}}
            \Else
                {\IlIf{$\Up[i].\Read{} = 1$}{\Return{\lose}}\label{ln:re-oblivious:readUp}}
    }
%        $i\gets i-1$\;
    \While{$i>1$}{
        $i\gets i-1$\;
        Choose $c_i'\in\{\heads,\tails\}$ at random such that $\Pr(c_i' = \heads) =
%          1/2^{b^{i-1}}$\;
          q_i$\;
%        Flip a biased coin $c$ that comes up \heads with probability $q_i$\;
            \uIf{$c_i' = \heads$}
                {$\Down[i]$.\Write{$1$}\label{ln:re-oblivious:writeDown}\;}
            \Else
                {\IlIf{$\Down[i].\Read{} = 1$}{\Return{\lose}}\label{ln:re-oblivious:readDown}}
    }
    \Return{\win}
  \end{method}
  \caption{A group election implementation for the r/w-oblivious adversary model.}
  \label{fig:GroupElect-RW-oblivious}
\end{classfigure}

We present a randomized group election implementation from registers, which has constant effectiveness and expected \maxstep{} complexity $O(\log\log k)$ in the r/w-oblivious adversary model.
This can be used to implement a TAS object with expected \maxstep{} complexity $O(\log\log k)$ against r/w-oblivious adversaries.

The group election implementation is given in Figure~\ref{fig:GroupElect-RW-oblivious}.
The algorithm consists of two phases, the \emph{backward sifting phase} and the \emph{forward sifting phase}.
The latter phase is similar to a sifting procedure used to eliminate processes in the TAS algorithm by Alistarh and Aspnes~\cite{AA_TAS_2011a}.
Their algorithm, however, is not adaptive.
To achieve that, the backward sifting phase runs essentially the same sifting procedure but in the opposite direction.

Two shared arrays of registers are used, one in each phase, namely, $\Up[1\dots\ell]$ and $\Down[1\dots\ell-1]$, where $\ell=\lceil\log_{b}\log n\rceil$ and $b=\frac32$.
All entries in both arrays are initially 0.

In the backward sifting phase, for each $i=1,2,\dots$, each process $p$ decides at random to either read register $\Up[i]$ or to write the value 1 to it.
The probability of writing decreases with $i$, more precisely, it is $q_i = 1/2^{b^{i-1}}$.
The phase ends for $p$ as soon as it has executed a read operation or has written to all registers of $\Up$.
If $p$ reads the value 1 on $\Up[i]$, it means that some other process has written to $\Up[i]$ before, and $p$ returns \lose immediately.
If $p$ reads 0 on $\Up[i]$, then it moves on to the forward sifting phase.
If $p$ writes to $\Up[i]$ instead, then it continues to the next element of $\Up$ if $i<\ell$, or if $p$ has already reached the end of array $\Up$, it moves on to the forward sifting phase.

Suppose that process $p$ reaches the forward sifting phase after reading the value 0 on register $\Up[i_p]$ for some index $i_p\in\{1,\dots,\ell\}$, or after writing the value 1 on register $\Up[i_p]$ for $i_p=\ell$.
Then, for each $i = i_p-1,i_p-2,\dots$,1, processes $p$ either reads register $\Down[i]$ or writes the value 1 to it.
As before, the decision is made at random and the probability of writing is $q_i$.
If $p$ reads the value 1, it returns \lose.
If $p$ writes to $\Down[i]$ or reads 0 from it, then $p$ continues to $\Down[i-1]$ if $i>1$, or $p$ returns \win if $i=1$.

%The idea behind the algorithm is as follows.
Let $k$ be the maximum number of processes participating in the group election.
Then with high probability no process accesses a register of array $\Up$ beyond the first $O(\log\log k)$ registers, because for larger indices $i$ the probability $q_i$ of writing to $\Up[i]$ is polynomially small in $k$.
This implies the $O(\log\log k)$ bound on the expected \maxstep{} complexity.
%\pwtodo{For that statement we need a given constant instead of $O(\log\log k)$.
%Perhaps say the probability is $1/k^\varepsilon$ for some large enough $\varepsilon$? Or perhaps just say the probability is polynomially small in $k$?}
The bound on the effectiveness is obtained as follows.
We have that the number $r_i$ of processes that move from the backward to the forward sifting phase after reading register $\Up[i]$ is in expectation bounded by $1/q_i$:
Each of those $r_i$ processes must read register $\Up[i]$ before any process has written to $\Up[i]$, and the probability of writing to that register is $q_i$.
We show by an inductive argument that the number $s_i$ of processes that access $\Down[i]$ and do not return \lose right after the operation is $O(1/q_i)$ in expectation, thus the number $s_1+r_1$ of processes that get elected is $O(1)$ in expectation.
The inductive argument goes as follows:
The number of processes that access $\Down[i]$ is %bounded by
$s_{i+1} + r_{i+1}$ (where $s_\ell$ is defined as the number of processes that write to $Up[\ell]$).
The expectation of $s_{i+1} + r_{i+1}$ is $O(1/q_{i+1})$, by the induction hypothesis and the earlier observation that $r_i$ is bounded by $1/q_i$.
The first write operation on $\Down[i]$ occurs in expectation after $1/q_i$ accesses, and after that only processes that write to $\Down[i]$ do not return \lose.
So in total the expected number of processes that do not return \lose after accessing $\Down[i]$ is at most $1/q_i$ plus the fraction $q_i (r_{i+1}+s_{i+1})$ of processes that write to $\Down[i]$.
A simple calculation bounds that by $O(1/q_i)$.

\begin{lemma}\label{lem:GroupElect-rw-oblivious}
  Figure~\ref{fig:GroupElect-RW-oblivious} gives a randomized implementation of a group election object with effectiveness at most $16$ and expected \maxstep{} complexity $O(\log\log k)$ against any r/w-oblivious adversary.
\end{lemma}

\begin{proof}
Let $M$ be an algorithm that uses the implemented group election object, and consider any execution of $M$.
%in which all processes participating in the group election finish their \elect{} call.
First we argue that not all \elect{} calls return \lose in the execution.
Suppose, towards a contradiction, that they all do.
Then each process reads the value 1 on some register $\Up[i]$ or some register $\Down[i]$, and the process returns $\lose$ immediately after that.
This implies that at least one process writes the value 1 to some register.
We argue that no process writes to any of the registers $\Down[i]$:
Otherwise, let $i_{\min}$ be the smallest index such that some process $p_{\min}$ writes the value 1 to $\Down[i_{\min}]$.
But then $p_{\min}$ does not return \lose at any point, because after writing to $\Down[i_{\min}]$, $p_{\min}$ may only read registers $\Down[j]$ for $j<i_{\min}$.
Thus, some process  must write to a register $Up[i]$.
Let $i_{\max}$ be the largest index such that some process $p_{\max}$ writes the value 1 to $\Up[i_{\max}]$.
But then $p_{\max}$ does not return \lose at any point, as after writing to $\Up[i_{\max}]$, $p_{\max}$ may only read registers $\Up[i]$ for $i>i_{\max}$, and registers $\Down[i]$, for $1\leq i\leq\ell-1$, none of which has value 1.

Next we bound the effectiveness of the implementation.
Let $A$ be some r/w-oblivious adversary, and let $\EE:=\EE_{M,A}$ be a random execution of algorithm $M$ scheduled by $A$.
Fix the prefix $\EE'$ of $\EE$ until the first process is poised to invoke \elect{}, and let $k := k_{max}^{M,\elect{}}(\EE')$ be the max-contention of \elect{} in $\EE$.
%Next we bound the effectiveness of the implementation.
%Let $A$ be some r/w-oblivious adversary, and $\EE$ a random execution of algorithm $M$ scheduled by $A$.
%Fix the max-contention $k$ of the implementation against $A$.
For $1\leq i\leq \ell$, let $r_i$ be the number of processes in $\EE$ that read register $\Up[i]$
%(in line~\ref{ln:re-oblivious:readUp})
before any process writes to it.
%(in line~\ref{ln:re-oblivious:writeUp}).
For $1\leq i\leq\ell-1$, let $s_i$ be the number of processes that either read register $\Down[i]$ before any process writes on it, or write on
register $\Down[i]$.
%(in line~\ref{ln:re-oblivious:writeDown})
%(in line~\ref{ln:re-oblivious:readDown})
%``survive'' after accessing register $\Down[i]$ (in lines~\ref{ln:re-oblivious:writeDown} or~\ref{ln:re-oblivious:readDown}), i.e., they either write on $\Down[i]$ or read the register before any process writes on it.
We also define $s_\ell$ to be the number of processes that write on $\Up[\ell]$.
%If a process accesses register $\Down[i]$ (in line~\ref{ln:re-oblivious:readDown} or \ref{ln:re-oblivious:writeDown}), then the preceding shared memory access must have either been on $\Down[i+1]$ in the process' preceding iteration of the while-loop, or, since $i<\ell$, it must have read $\Up[i+1]=0$ in line~\ref{ln:re-oblivious:readUp} and then broken out of the repeat-until loop.
The total number of processes that access register $\Down[i]$, for $1\leq i\leq \ell-1$, is then at most\footnote{We say `at most' instead of `exactly' because we do not require that all processes finish their \elect{} call in $\EE$.}
$r_{i+1} + s_{i+1}$, and the number of processes that get elected in the group election is
$r_{1} + s_{1}$.
%\pwtodo{Suppose $i=\ell-1$. Then  a process may have previously written to $\Up[i+1]=\Up[\ell]$ and broken out of the repeat-until loop. So in this case perhaps $r_{i+1}+s_{i+1}$ is not an upper bound for the number of processes that access register $\Down[i]$? So perhaps $r_\ell$ needs to be the number of processes that access $\Up[\ell]$?}
%\gtodo{see comment above}
We will show that $\Exp[r_{1} + s_{1}] \leq 16$.

Since adversary $A$ is r/w-oblivious, it does not know whether a process poised to access a shared register will read or write to that register.
We can thus assume that right after the last step of $\EE'$, we perform for each register $\Up[i]$ and each register $\Down[i]$ a series of $k$ independent coin flips with heads probability $q_i$, and that the $j$-th process to subsequently accesses that register uses the $j$-th coin flip value in the series to decide whether it should read or write on the register.
We observe that once these series of coin flips have been fixed, the values of all random variables $r_i$ and $s_i$ are completely determined by the number $k'\leq k$ of processes that invoke \elect{},
provided that all these $k'$ processes finish their \elect{} call.
(In particular, $r_i$ and $s_i$ do not depend on the order in which the $k'$ processes are scheduled to take steps.)
Moreover, if not all $k'$ \elect{} calls are executed to completion, then for each $1\leq i\leq \ell$, $r_i$ and $s_i$ are smaller or equal than the corresponding values if all  $k'$ calls were executed to completion.
%(assuming the same series of coin flips for each register in both executions).

It follows that instead of the schedule determined by adversary $A$, we can consider a schedule with the following convenient properties:
Exactly $k$ of processes call the implemented \elect{} method and all processes finish their call;
for each $1\leq i< \ell$, all operations on register $\Up[i]$ are scheduled before any operation on $\Up[i+1]$; for each $1 < i \leq \ell-1$, all operations on $\Down[i]$ are scheduled before any operation on $\Down[i-1]$; and all operations on array $\Up$ are scheduled before any operation on array $\Down$.
Let $R_i$ and $S_i$ denote the same quantities as $r_i$ and $s_i$ but for a schedule as described above.
Then $R_i \geq r_i$ and $S_i \geq s_i$, if the same series of coin flips are used  for each register under both schedules, and thus
\[
    \Exp[R_{1} + S_{1}] \geq \Exp[r_{1} + s_{1}].
\]

We now bound  $\Exp[R_{1} + S_{1}]$.
For that we no longer assume that coin flips are fixed in advance.

For each $1\leq i \leq \ell-1$, the values of $R_{i+1}$ and $S_{i+1}$ are determined before the first process accesses register  $\Down[i]$.
It follows that
\[
    \Exp[S_i \mid R_{i+1}, S_{i+1}] \leq 1/q_i + q_i(R_{i+1} + S_{i+1}),
\]
where the term $1/q_i$ accounts for the processes that read $\Down[i]$ before any process writes on $\Down[i]$, and the term $q_i(R_{i+1} + S_{i+1})$ accounts for the processes that write on $\Down[i]$.
Taking the unconditional expectation yields
\begin{equation}
    \label{eq:ESi}
    \Exp[S_i]
    \leq
    1/q_i + q_i(\Exp[R_{i+1}] + \Exp[S_{i+1}])
%    \\&
    \leq
    1/q_i + q_i(1/q_{i+1} + \Exp[S_{i+1}])
    .
\end{equation}
%Further, for each $1\leq i\leq \ell$, we have
%$
%    \Exp[R_{i}] \leq 1/q_{i}.
%$
%From the last two inequalities, it follows
%\[
%    \Exp[S_i] \leq 1/q_i + q_i (1/q_{i+1} +\Exp[S_{i+1}]).
%\]

We now show by induction on $i = \ell,\ell-1,\dots,1$ that
\begin{equation}
    \label{eq:IH}
    \Exp[S_i] \leq 7/q_i.\tag{IH}
\end{equation}
Recall that $\ell = \lceil\log_{b}\log n\rceil\geq \log_{b}\log n$, and $b=3/2$. We also have $q_i = 1/2^{b^{i-1}}$ and thus $q_{i+1} = q_i^b$.
For the base case of $i=\ell$, we have that $S_\ell$ is the number of processes that write to $Up[\ell]$, thus
\[
    \Exp[S_\ell]
    \leq k  q_\ell
    =
    \frac{k  q_\ell^2}{q_\ell}
    \leq
    \frac{k  \big(1/2^{b^{\log_b\log n-1}}\big)^2}{q_\ell}
    =
    \frac{k /n^{4/3}}{q_\ell}
    \leq
    \frac{n^{-1/3}}{q_\ell}
    <
    \frac{7}{q_\ell}.
\]
For $i<\ell$, we obtain from~\eqref{eq:ESi} that
\begin{align*}
    \Exp[S_i]
    &\leq
    1/q_i + q_i (1/q_{i+1} + \Exp[S_{i+1}])
    \\&
%    \stackrel{\text{(IH)}}{\leq}
    \leq
    1/q_i + q_i(1/q_{i+1} + 7/q_{i+1}), \text{\quad by \eqref{eq:IH}}
    \\&
%     \qquad\text{(by the induction hypothesis)}
%     \\&
    =
    1/q_i + 8q_i/q_{i+1}
    \\&
    =
    1/q_i + 8q_{i}/q_i^b
    \\&
    =
    (1/q_i) (1 + 8q_{i}^{2-b})
    \\&
    \leq
    (1/q_i) (1 + 8q_1^{2-b})
%    \\&
%    =
%    (1/q_i) (1 + 8(1/2)^{1/2})
    \\&
%    <
%    (1/q_i)\cdot 6.6
%    \\&
    <
    (1/q_i)\cdot 7.
\end{align*}
This completes the inductive proof that $\Exp[S_i] \leq 7/q_i$.
Applying this inequality, for $i=1$, we obtain
\[
    \Exp[R_{1} + S_{1}]
    \leq
    1/q_1 + 7/q_1
    =
    16.
\]
Therefore, the effectiveness of the implemented group election is $\Exp[r_{1} + s_{1}] \leq \Exp[R_{1} + S_{1}]\leq 16$.

It remains to bound the expected \maxstep{} complexity of the implementation.
Let $i^\ast$ be the maximum index $i$ such that some process accesses register $\Up[i]$ in execution $\EE$.
Then the maximum number of shared memory operations by any process is at most $2i^\ast - 1$.
We have that $\Pr(i^\ast\geq i)$ is bounded by the expected number of processes that access $\Up[i]$, and this is bounded by $kq_i$.
%Recall that $b=\frac32$ as defined in Figure~\ref{fig:GroupElect-RW-oblivious}.
Thus, for $\lambda := \ceil{\log_{b}\log k}$, we have
\begin{align*}
    \Exp[i^\ast]
    &=
    \sum_{i\geq 1} \Pr(i^\ast\geq i)
    \\&
    \leq
    \lambda + \sum_{i \geq  \lambda+1} \Pr(i^\ast\geq i)
    \\&
    \leq
    \lambda + \sum_{i \geq  \lambda+1} kq_i
    \\&
    \leq
    \lambda + kq_{\lambda+1} \sum_{i \geq  0} q_{\lambda+1}^{b^{i}-1}.
\end{align*}
Since $q_{\lambda+1} = 1/2^{b^{\lambda}}\leq 1/k$ and
$\sum_{i \geq  0} q_{\lambda+1}^{b^{i}-1}\leq \sum_{i \geq  0} q_1^{b^{i}-1}
%= \sum_{i \geq  0} 1/2^{(3/2)^{i}-1}
< 3$, it follows that
$
    \Exp[i^\ast]
    \leq
    \lambda + 4.
$
Hence, the expected \maxstep{} complexity is at most $2\Exp[i^\ast] - 1 \leq 2(\lambda + 4) - 1 = 2\ceil{\log_{b}\log k} + 7$.
This completes the proof of Lemma~\ref{lem:GroupElect-rw-oblivious}.
\end{proof}

\begin{theorem}
    \label{thm:rw-oblivious}
    There is a randomized implementation of a TAS object from $\Theta(n)$ registers with expected \maxstep{} complexity $O(\log\log k)$ against any r/w-oblivious adversary.
\end{theorem}
\begin{proof}
We consider the TAS implementation of Figure~\ref{fig:generalLE}, and use the algorithm in Figure~\ref{fig:GroupElect-RW-oblivious} to implement the group election objets $G[j]$, for $1\leq j \leq 16$.
For $16 < j \leq n$, we let $G[j]$ by a trivial group election object, where all participating processes get elected and the \maxstep{} complexity is zero.
From Lemma~\ref{lem:GroupElect-rw-oblivious}, the group election objects $G[j]$, for $1\leq j\leq 16$, have effectiveness at most $16$ and expected \maxstep{} complexity   $O(\log\log k)$ against any r/w-oblivious adversary.
Theorem~\ref{thm:groupelect->leaderelect} then implies that the resulting TAS algorithm has expected \maxstep{} complexity $O(16\cdot \log\log k)$ against any location-oblivious adversary.
Moreover the algorithm uses $\Theta(n)$ registers, as each of the first $16$ group election objection uses $O(\log\log n)$ registers, and the remaining trivial group election objects do not use any registers.
%
%Consider the TAS object obtained by using the algorithm in Figure~\ref{fig:GroupElect-RW-oblivious} to implement the group election objects of the TAS algorithm given in Figure~\ref{fig:generalLE}.
%From Lemma~\ref{lem:GroupElect-rw-oblivious}, the group election objects have constant effectiveness and expected \maxstep{} complexity $O(\log \log k)$ against any r/w-oblivious adversary.
%From Lemma~\ref{lem:groupelect->leaderelect} the resulting TAS object has expected \maxstep{} complexity $O(\log\log k)$ against any r/w-oblivious adversary.
%The total number of registers used is $O(n\log\log n)$, but  can be reduced to $O(n)$ in the same way as explained in the proof of Theorem~\ref{thm:location-oblivious}, by removing all but the first $O(\log n  \log\log n)$ group election and doorway objects.
\end{proof}

\section{Linear-Space TAS for Strong Adaptive Adversaries}

\label{sec:ratrace}

We present a TAS implementation from $\Theta(n)$ registers that has \maxstep{} complexity $O(\log k)$ both in expectation and w.h.p.\ (i.e., with probability ${1-1/k^{\Omega(1)}}$), against any strong adaptive adversary.
Our implementation is a variant of the \RatRace algorithm proposed by Alistarh et al.~\cite{AAGGG2010a}, which has the same \maxstep{} complexity but uses $\Theta(n^3)$ registers.

\begin{theorem}\label{thm:ratrace}
  There is a randomized implementation of a TAS object from $\Theta(n)$ registers with \maxstep{} complexity $O(\log k)$, both in expectation and w.h.p., against any strong adaptive adversary.
%  The \maxstep{} complexity is $O(\log k)$ also w.h.p.
\end{theorem}

Before we prove Theorem~\ref{thm:ratrace}, we give an overview of the original \RatRace algorithm.
%The remainder of this section is devoted to the proof of this theorem.
To simplify exposition, throughout this section we treat $\log n$, $n/\log n$, and $\log\log n$ as integers.
It is easy to accommodate the calculations for the case that this is not true, by rounding appropriately.
%We make this assumption to avoid complicating notation and calculations unnecessarily, which would distract from the core ideas.

\paragraph{Overview of RatRace.}
\RatRace~\cite{AAGGG2010a} uses two shared memory data structures, a \emph{primary tree} and a \emph{backup grid}.
The primary tree is a perfect binary tree of height $3\log n$, where each node $v$ stores a randomized splitter
%(\RSplitter)
object $S_v$, and a randomized 3-process TAS object $T_v$.
The latter can be implemented from two 2-process TAS objects. %\TASII.

Each process $p$ starts at the root of the primary tree and moves downwards towards the leaves.
The process goes through the splitters at the nodes it visits along the way, until it stops at a splitter, or ``falls off'' the bottom of the tree (which happens only with low probability).
If $p$ turns left or right at a splitter $S_v$, then it moves respectively to the left or right child of $v$, provided $v$ is not a leaf.
If $v$ is a leaf, $p$ moves to the backup grid as explained below.
If $p$ stops at $S_v$ then it stops moving downwards, and starts to move upwards towards the root, along the same path.
At each node $u$ in the path to the root, $p$ tries to win the TAS on object $T_u$.
If $p$ loses that TAS, it immediately loses the implemented TAS.
Otherwise, it moves to the parent of $u$ in the tree.
The process that wins the TAS at the root competes against the winner at the backup grid.

%For a node $(i,j)$, we use the convention that its left child is node $(i+1,j)$, if $i<n$, and its right child is node $(i,j+1)$, if $j<n$.

The backup grid is an $n\times n$ square grid, where each node $v = (i,j)\in\{1,\ldots,n\}^2$ stores a \emph{deterministic} splitter object,
%(\Splitter)
and also a randomized 3-process TAS object as before.
We define the left and right children of node $(i,j)$ at the grid to be nodes $(i+1,j)$ and $(i,j+1)$, respectively.
Each process that falls off the primary tree starts at  node $(1,1)$, and proceeds in a similar way as in the primary tree:
At each node the process goes through the splitter, moving to the child as indicated by the direction to which the process turns at the splitter, until it stops at some splitter.
Then, the process tries to move back to node $(1,1)$ along the same path, by winning all the TAS in the nodes along the way.
The properties of deterministic splitters guarantee that the process wins a splitter before it falls off the grid.

The winner of the TAS at node $(1,1)$ of the backup grid, and the winner of the TAS at the root of the primary tree participate in a randomized 2-process TAS, which determines the  winner of \RatRace.

%The above TAS implementation is not linearizable.
To ensure linearizability, a doorway object is used such that only processes that pass through the doorway participate in the above algorithm, whereas processes that are deflected lose immediately.

%\subsection{Improving the Space Complexity}\label{sec:Rat_Race_Space_Improvement}
\paragraph{Reducing the Space Complexity (Proof of Theorem~\ref{thm:ratrace}).}
\RatRace requires $\Theta(2^{3\log n}) = \Theta(n^3)$ registers for the primary tree of height $3\log n$, and $\Theta(n^2)$ registers for the backup $n\times n$ grid.
Next we show how to reduce this space complexity, without increasing the \maxstep{} complexity.

We use a data structure,
which we call an \emph{elimination path},
that is similar to the backup grid but uses fewer registers.
%; we call it an \emph{elimination path}.
An \emph{elimination path of length $\ell$} is an $\ell$-node path where each node $i\in\{1,\ldots,\ell\}$ stores a deterministic splitter $S_i$, and a randomized 2-process TAS object $T_i$.
The possible outcomes for a process accessing an elimination path is to \emph{win}, \emph{lose}, or \emph{fall off} the path.
A process $p$ {enters} the elimination path at node $i=1$, and moves towards node $\ell$, going through splitter $S_i$ at each node $i$ it visits.
If $p$ turns left at $S_i$, then it \emph{loses} and takes no more steps.
If it turns right, then it moves to the next node, $i+1$, if $i<\ell$, whereas if $i=\ell$, $p$ \emph{falls off} the path
and takes no more steps in the path.
%(We show below that this won't happen if the path is long enough.)
Last, if $p$ stops at $S_i$, then it starts moving back towards node~1.
From node $i>1$, it moves to $i-1$ if it wins the TAS on $T_i$, otherwise, it loses and stops.
The \emph{winner} of the elimination path is the winner of~$T_1$.

%The following simple observation holds.

With some slight modifications, the TAS algorithm in Figure~\ref{fig:generalLE} implements an elimination path of length $n$.
More precisely, we remove line~\ref{ln:LE:doorway}, where the process accesses the doorway, and replace line~\ref{ln:LE:elect}, where the process participates in a group election, with the statement:
\textbf{if} $i>n$ \textbf{return} \falloff.
The process wins (loses) if the return value is 0 (respectively, 1).
%We use this observation in the proof of Claim~\ref{clm:elimination-path} below,
%
%Note that an elimination path of length $n$ is the same as the TAS algorithm in Figure~\ref{fig:generalLE}, if we remove lines~\ref{ln:LE:doorway} and~\ref{ln:LE:elect} of the algorithm, where processes access a doorway and participate in a group election, respectively.

The next lemma summarizes the main properties of an elimination path.

\begin{lemma}
    \label{lem:elimination-path}
    At most one process wins in an elimination path, and not all processes that access the elimination path lose.
    If $k\leq \ell$ processes access an elimination path of length $\ell$, then no process visits a node with index $j>k$, and no process falls off. %the end of the path.
%
%    In any execution in which processes access an elimination path of length $\ell$, at most one process wins and not all processes lose. If at most $k\leq \ell$ processes participate, then no process visit any of the nodes $i>k$, and in particular, no process falls off.
%
%    If at most $\ell$ processes enter an elimination path of length $\ell$, then no process falls off.
\end{lemma}
\begin{proof}
The properties that at most one process wins and not all processes lose follow from the  same properties of the TAS implementation in Figure~\ref{fig:generalLE}.
For the second part of the lemma, we have that at each splitter, not all processes can turn right.
Hence, if at most $k$ processes enter the elimination path, then at most $k-i$ processes turn right at splitter $S_i$, for $i\leq k$.
This implies that no process visits a node with index $j>k$, and that no process falls off the end of the path.
\end{proof}

%\begin{claim}\samepage
%    \label{clm:elimination-path}
%    In an elimination path, at most one process wins, and not all processes lose.
%    Also if at most $\ell$ processes enter, where $\ell$ is the length of the path, then no process falls off.
%\end{claim}
%
%\begin{proof}
%At most one process wins as that process must win $T_1$.
%Not all processes lose, because at least one of the processes that go through a splitter does not turn left (and thus does not lose after the splitter), and also not all processes that access the TAS objects can lose.
%Finally, since at each splitter not all processes can turn right, if at most $\ell$ processes enter the elimination path, then at most $\ell-i$ processes turn right at splitter $T_i$ of node $i$, and thus no process falls off.
%\end{proof}

To reduce the space complexity of the \RatRace algorithm, the first modification we make is to replace the backup grid by a \emph{backup elimination path} $B$ of length $n$.
Lemma~\ref{lem:elimination-path} implies that $B$ has the same properties as the backup grid against a strong adaptive adversary.
%Similar to the backup grid, it guarantees that no process falls off (by Claim~\ref{clm:elimination-path}), that at most one process wins, that not all processes lose, and that each process goes through at most $O(n)$ splitters and participates in at most $O(n)$ 2-process TAS.
Unlike the backup grid however, $B$ requires only $\Theta(n)$ registers.

A second modification is that we replace the primary tree of height $3\log n$, by a data structure consisting of a smaller primary tree, of height  $\log n - \log\log n$, and $n/\log n$ elimination paths $P_i$ of length $4\log n$, where $1\leq i\leq n/\log n$.
Note that we have as many elimination paths as the leaves of the primary tree.
The total number of registers required is $\Theta(2^{\log n-\log\log n} + (4\log n)\cdot n/\log n) = \Theta(n)$.
The primary tree is used in the same way as before, but now any process that falls off moves to one of the elimination paths, instead of the backup grid.
More precisely, a process that falls off the $i$-th leaf moves to elimination path $P_i$.
The winner at each $P_i$ (if there is one) moves back to the primary tree, at leaf $i$, and from there it tries to reach the root as in the original \RatRace algorithm.
Any process that falls off a path $P_i$ moves to the backup elimination path $B$.
Finally, as before, the winner of $B$ and the winner of the primary tree participate in a 2-process TAS to determine the winner of the implemented TAS.

Consider a random execution of an algorithm that uses the above TAS implementation, scheduled by a strong adaptive adversary.
Fix the prefix of this execution until the first process is poised to invoke the implemented TAS, and suppose the max-contention is $k$.

%Suppose at most $k$ processes participate in the implemented TAS.
If $\log k \leq(\log n-\log\log n)/3$, then a bound of $O(\log k)$ on the expected \maxstep{} complexity, and also on the \maxstep{} complexity w.h.p., follows from the analysis of the original \RatRace~\cite{AAGGG2010a}.
%(as w.h.p.\ no process falls of the primary tree).

In the following we assume that
$\log k > (\log n-\log\log n)/3$.
We use the next simple lemma, which implies that w.h.p.\ the number of processes that enter each elimination path $P_i$ is not greater than its length.

\begin{lemma}
    \label{lem:procETs}
    With probability at least $1 - 1/n$, each leaf node in the primary tree is visited by at most $4\log n$ processes.
\end{lemma}
\begin{proof}
The number of processes that visit a given leaf node is stochastically dominated by the number of balls that fall in a given bin in the standard bins-and-balls model, with  $n$ balls and $n/\log n$ bins. In this model each ball is placed in a bin chosen independently and uniformly at random.
The domination follows because we can assume each process $p$ comes with an independent and uniform random bit string of length $\log n - \log\log n$.
If $p$ goes through a randomized splitter in a node at distance $i-1$ from the root, and does not stop at that splitter, then the $i$-th bit in the bit string determines whether $p$ will turn left or right at the splitter.
Hence, the random bit string uniquely determines the leaf that $p$ will reach, if it does not stop at any splitter along the way.

For $1\leq i\leq n$, let $X_i$ be the 0/1 random variable that is 1 if and only if the $i$-th ball falls in some fixed bin $b$.
Let $X = X_1+\dots+X_n$ be the total number of balls that fall in $b$.
Then $\Exp[X]=\log n$, and by a standard Chernoff bound,
%(see Theorem~\ref{thm:Chernoff} in the appendix),
stated as Theorem~\ref{thm:Chernoff} below, we obtain
\[
    \Pr(X > 4 \log n)
    \leq
    e^{-\frac{3^2\log n}{2(1+1)}}
    <
    n^{-2}.
\]
Therefore, the same $n^{-2}$ upper bound applies to the probability that more than $4\log n$ processes visit a given leaf node.
Then by a union bound, the probability that the maximum number of visits at any of the $n/\log n$ leaves exceeds $4\log n$ is at most $n^{-1}/\log n$.
\end{proof}

The following Chernoff Bound, used in the proof above, can be found in~\cite[Theorem~2.3(b)]{McDiarmid1998}.

\begin{theorem}[Chernoff Bound]
\label{thm:Chernoff}
  Let $X_1,X_2,\dots,X_n$ be independent random variables with $0\leq X_i\leq 1$, for each $i\in\{1,\dots,n\}$, and let $X=X_1+\dots+X_m$ and $\mu=\Exp[X]$.
  Then for any $\delta>0$,
  \begin{displaymath}
    \Pr\bigl(X\geq (1+\delta)\mu\bigr)\leq e^{-\frac{\delta^2\mu}{2(1+\delta/3)}}.
  \end{displaymath}
\end{theorem}

From Lemma~\ref{lem:procETs}, we have that  w.h.p.\ no more than $4\log n$ processes enter any single elimination path $P_i$, and thus w.h.p.\ no process enters the backup elimination path $B$, by Lemma~\ref{lem:elimination-path}.
%Since there are $n/\log n$ elimination paths $P_i$, each of which corresponds to a unique set of $n/\log n$ leaves, by the union bound we have with probability at least $1-1/n\log n$, that no more than $4\log n$ processes enter any single $P_i$.
%Thus, by Claim~\ref{clm:elimination-path}, it follows w.h.p.\ that no process enters the backup elimination path $B$.
If no process enters $B$, then each process traverses at most a path of length $\log n-\log\log n$ in the primary tree (from the root to a leaf), and at most one of the elimination paths $P_i$ of length $4\log n$.
Therefore, each process goes through at most $O(\log n)$ splitters, and participates in at most $O(\log n)$ 3-process TAS objects.
It follows that the \maxstep{} complexity is bounded by $O(\log n) = O(\log k)$ w.h.p.
Since w.h.p.\ no process reaches $B$, and the maximum number of steps a process takes at $B$ is $O(n)$ w.h.p., it follows that the expected  \maxstep{} complexity is bounded by $O(\log k)$, as well.
This completes the proof of Theorem~\ref{thm:ratrace}.
\qed

%\clearpage

%\section{Adversary Independence}
\section{Combining TAS Algorithms for Different Adversaries}
\RatRace and its linear-space variant presented in Section~\ref{sec:ratrace} achieve logarithmic %adaptive
\maxstep{} complexity in the strong adaptive adversary model.
These algorithms do not benefit from weaker adversaries, as their expected \maxstep{} complexity is still logarithmic even in the oblivious adversary model.
On the other hand, the TAS implementations in Section~\ref{sec:algos-oblivious}, which are more efficient against weaker adversaries, exhibit poor performance in the strong adaptive adversary model, having linear expected \maxstep{} complexity. %$\Theta(k)$.
In this section we describe how one can combine any of the implementations in Section~\ref{sec:algos-oblivious} with \RatRace, to obtain a TAS object that has the expected \maxstep{} complexity of \RatRace against any strong adaptive adversary, and the expected \maxstep{} complexity of the corresponding algorithm in Section~\ref{sec:algos-oblivious} in the weaker adversary model.
\begin{theorem}\samepage\label{thm:combine-algos}
    For any randomized TAS implementation~$\Imp$, there is a randomized TAS implementation \Comb that has the following properties:
    \begin{enumerate}[label=(\alph*)]
     \item\label{item:a}
     If $f$ is a non-decreasing function such that the %adaptive
     expected \maxstep{} complexity of $\Imp$ is at most $f(k)$ against any location-oblivious (or r/w-oblivious) adversary, then \Comb has %adaptive
     expected \maxstep{} complexity $O\bigl(f(k)\bigr)$ against any location-oblivious (respectively r/w-oblivious) adversary;
     \item\label{item:b}
     \Comb has %adaptive
     expected \maxstep{} complexity $O(\log k)$  against any strong adaptive adversary; and %both in expectation and w.h.p.; %and
     \item\label{item:c}
     The space complexity of \Comb is $\Theta(n)$ plus the space complexity of $\Imp$.
    \end{enumerate}
\end{theorem}

Combining Theorem~\ref{thm:combine-algos} with Theorems~\ref{thm:location-oblivious} and~\ref{thm:rw-oblivious}, yields the following result.
\begin{corollary}
    There are randomized implementations of TAS objects from $\Theta(n)$ registers with %adaptive
    expected \maxstep{} complexity $O(\log^\ast k)$ or $O(\log\log k)$
    against any location-oblivious adversary or any r/w-oblivious adversary, respectively,
    and with expected \maxstep{} complexity $O(\log k)$ against any strong adaptive adversary.
%    There is a TAS implementation that has an expected \maxstep{} complexity of $O(\log^\ast k)$ (respectively $O(\log\log k)$) against the location-oblivious (respectively r/w-oblivious) adversary, and a \maxstep{} complexity of $O(\log k)$, both in expectation and w.p.\ $1-1/k$, against the adaptive adversary.
%    The space complexity is $\Theta(n)$.
%\pwtodo{Check this later, when earlier chapters have been revised.}
\end{corollary}

\subsection{Proof of Theorem~\ref{thm:combine-algos}}

\subsubsection{Implementation}
We present a TAS implementation, \Comb, which
achieves the step and space complexities stated in Theorem~\ref{thm:combine-algos}.
%The implementation uses a variant of \RatRace, in which the initial doorway is removed.
Each process first enters a doorway $D$, and the processes that get deflected lose immediately.
A process that passes through $D$, then runs both \Imp and a variant of \RatRace, in parallel.
The only difference of the \RatRace variant used from the original \RatRace is that its initial doorway is removed.
 %alternating steps of \RatRace and $\Imp$.
More precisely, after passing through $D$, each process executes a step of \Imp in every odd step, and a step of \RatRace (without doorway) in every even step.
%But instead of having two doorways, one for $\Imp$ and one for $\RatRace$,
%we now assume that there is one common doorway $D$ that processes have to pass through before they start alternating steps between the doorway-less
%variants of $\Imp$ and $\RatRace$.
%(If a process gets deflected on $D$, it immediately loses.)

A natural way to combine the two interleaved executions would be that each process takes steps until it either wins or loses in one of the two algorithms; if it loses it also loses in the combined implementation, and if it wins in one of the two algorithms it competes against the winner of the other algorithm.
This approach, however, could yield an execution in which no process wins.
For instance, suppose that $\Imp$ is also an instance of \RatRace.
In an execution in which only two processes, $p$ and $q$, participate, process $p$ might loses against $q$ on one of the 2- or 3-process TAS objects
in the first instance of \RatRace, and at the same time $q$ may lose against $p$ on a TAS object in the second instance of \RatRace; thus all processes lose.

To solve this problem we impose the rule that if a process loses in $\Imp$ at a point when it has already stopped at some splitter object in \RatRace, then the process continues to execute \RatRace.
More precisely, we use the rules below to combine the two executions, with the help of an auxiliary 2-process TAS object $T_{top}$.
\begin{enumerate}[label=(C\arabic*)]
%\advance\itemsep-1ex
    \item \label{C1}
        If a process {wins} either \RatRace or $\Imp$, then it stops taking steps in the other algorithm, and tries to win $T_{top}$; if it wins $T_{top}$ then it wins the implemented TAS object, otherwise it loses.\label{rule:win}
    \item \label{C2}
        If a process {loses} \RatRace then it stops taking steps in $\Imp$, and it loses the implemented TAS object.\label{rule:lose-ratrace}
    \item \label{C3}
        If a process {loses} $\Imp$ while it has a pending \split{} call on a (randomized or deterministic) splitter of \RatRace, then it keeps taking steps in \RatRace, until its pending \split{} operation completes.
        Once it has no more pending \split{} operation it does one of the following:
        \begin{enumerate}[label=(C3\alph*)]
          \item \label{C3a}
            If it has not yet stopped at any of the splitter objects in \RatRace, then it stops taking steps in $\RatRace$, and it loses the implemented TAS object.\label{rule:lose-A}
          \item \label{C3b}
        If it has already stopped at one of the splitter objects in \RatRace, then it continues taking steps in \RatRace until \RatRace finishes, and it either wins or loses \RatRace.
        If it wins \RatRace, then it proceeds as in \ref{C1}; otherwise it loses the implemented TAS.
%        For uniformity reasons, the process will continue executing \RatRace in each even step, while each odd step is a ``dummy step'', i.e., an arbitrary shared memory operation that does not affect any other process (e.g., a read of some register).
%        (Precisely, the model requires that every other step is a coin flip, so odd steps alternate between dummy shared memory steps
%        and coin flip steps whose results are being ignored.)
%
       \end{enumerate}
\end{enumerate}

We now prove that \Comb is a correct (linearizable) TAS implementation, and then we show that it satisfies properties~\ref{item:a}--\ref{item:c} of Theorem~\ref{thm:combine-algos}.

\subsubsection{Correctness}
A process accesses $T_{top}$ if and only if it wins either \RatRace or $\Imp$.
It follows that at most two processes can execute the \TAS{} operation on $T_{top}$,
one that won \RatRace and one that won \Imp,
and thus at most one process can win \Comb. %in any execution.
In the following we show that in any execution in which all processes complete their \TAS{} call, at least one process wins (therefore exactly one process wins).
Due to the initial doorway $D$, linearizability follows from exactly the same arguments as for the algorithm in Section~\ref{sec:le-from-ge} (see the correctness proof of Theorem~\ref{thm:groupelect->leaderelect}).

%In the following we will show that in every execution in which all processes complete their TAS operation, exactly one process wins.
%Then due to the doorway $D$, linearizability follows with exactly the same arguments as for the construction given in Section~\ref{sec:le-from-ge} (see the correctness proof of Lemma~\ref{lem:groupelect->leaderelect}).
%
%Note that a process accesses $T_{top}$ if and only if it wins either \RatRace or $\Imp$.
%It follows that at most two processes can execute the TAS operation on $T_{top}$ and thus at most one process can win \Comb.
%Thus, it suffices to show that not all processes lose \Comb. %, i.e., if all participating processes are scheduled to take sufficiently many steps, then at least one of them wins either \RatRace or $\Imp$.

For the purpose of a contradiction, consider an execution $\EE$ in which all participating processes take sufficiently many steps to finish \Comb, and they all lose \Comb.
Let $Q$ be the set of processes that stop at some \RatRace splitter in $\EE$.

First suppose that $Q$ is empty.
Then no process wins or loses \RatRace, as otherwise it would have first stopped at some splitter, and thus it would be in $Q$.
Hence, by \ref{C1}--\ref{C3}
%\gtodo{maybe say instead that all processes continue to take steps until they either win or lose Imp}
%\pwtodo{Not sure what you mean}
all processes
execute $\Imp$ to completion, and
either win or lose $\Imp$ eventually.
Then, by the assumption that $\Imp$ is a correct TAS algorithm, exactly one process wins \Imp, and this process also wins $T_{top}$, since no process wins \RatRace, and hence no other process participates in a \TAS{} operation on $T_{top}$.
This contradicts the assumption that all processes lose \Comb.

Now suppose that $Q$ is not empty, that is, in execution $\EE$ at least one process stops at a splitter of \RatRace.
By the assumption that all processes lose \Comb in $\EE$, there is no process that wins either $\Imp$ or \RatRace (otherwise that process would execute a \TAS{} operation on $T_{top}$ and some process would win $T_{top}$, and thus \Comb).
In particular, no process in $Q$ wins \RatRace, and since by \ref{C3b}, each $q\in Q$ does not stop taking steps in \RatRace even after losing $\Imp$, it must lose \RatRace at some point.
Hence, each process in $Q$ loses a \TAS{} operation on some 2- or 3-process TAS object of \RatRace.
Recall that the TAS objects used by \RatRace are arranged in a rooted tree (where we consider the elimination paths as part of the tree).
Whenever a process wins a non-root TAS object $T$ of that tree, it continues to the parent of $T$.
Among all TAS objects on which processes in $Q$ lose, let $T^\ast$ be one that is closest to the root.
Then there must be a process $q\in Q$ that wins $T^\ast$, so $q$ ascends to the parent of $T^\ast$.
Then $q$ must lose on some other TAS object closer to the root than $T^\ast$, which contradicts the definition of $T^\ast$.

\subsubsection{Complexity} %Analysis.}
The linear space complexity of \Comb claimed in part~\ref{item:c} of Theorem~\ref{thm:combine-algos} follows immediately from the construction and our \RatRace implementation given in Section~\ref{sec:ratrace}, which uses $\Theta(n)$ registers (Theorem~\ref{thm:ratrace}).
We now analyze the expected \maxstep{} complexity of \Comb.

\paragraph{High Level Idea.}
We first describe the general idea for bounding the expected \maxstep{} complexity of \Comb, ignoring some of the subtleties that arise in the detailed analysis to follow.
We relate the expected \maxstep{} complexity of \Comb to the expected \maxstep{} complexity of \RatRace and \Imp, respectively, depending on what adversary model is used.
Note that the 2-process TAS object $T_{top}$ has constant expected \maxstep{} complexity even against a strong adaptive adversary, so it does not affect the asymptotic \maxstep{} complexity of \Comb.
Recall that during \Comb processes alternate between steps of \RatRace and \Imp until one of those two algorithms terminate, and if \RatRace terminates first, then the calling process also terminates its \Imp call (but not necessarily the other way around).
Therefore, the asymptotic max-step complexity of \Comb is dominated by that of \RatRace.
Hence, if a random execution of $k$ processes calling \Comb is scheduled by a strong adaptive adversary, then the maximum number of steps any process takes is $O(\log k)$, which is the upper bound for \RatRace as stated in Theorem~\ref{thm:ratrace}.

Now suppose such a random execution is scheduled by a location-oblivious or r/w-oblivious adversary.
It suffices to show that the expected maximum number of steps any process devotes to \RatRace during \Comb is bounded asymptotically by the expected maximum number of steps any process devotes to \Imp.
A process can devote more steps to \RatRace than to \Imp only if, by the time it finishes \Imp, it has either already stopped at a splitter in \RatRace, or it has a pending \split{} call that will return \Stop.
Hence, it suffices to consider processes that stop at \RatRace splitters.
Suppose a process stops at a \RatRace splitter in its $i$-th \split{} operation.
Since the process alternates between \RatRace and \Imp steps prior to its last \split{} operation, and each \split{} operation takes a constant number of steps, until finishing its $i$-th \split{} operation the process devotes $\Theta(i)$ steps to \RatRace and $\Theta(i)$ steps to \Imp.
In the remainder of its \RatRace execution, the process executes at most $i+1$ \TAS{} calls on 2- or 3-process TAS objects (one for each splitter it went through previously, in addition to $T_{top}$).
The number of steps for each such \TAS{} call is bounded by a geometrically distributed random variable.
Using Chernoff Bounds, we show that with probability $1-1/4^i$ the process needs only $O(i)$ steps for its at most $i+1$ \TAS{} calls to finish \RatRace after stopping at the $i$-th splitter.
Due to the arrangements of splitters in a primary tree and elimination paths, at most $2^i$ processes can stop after their $i$-th \split{} operation.
Thus, by a union bound applied to all processes stopping at the $i$-th splitter they go through, with probability exponentially close to $1$, all these processes need only $O(i)$ steps to finish \RatRace.
To summarize: all processes that stop at their $i$-th splitter devote $\Omega(i)$ steps of \Comb to \Imp, $\Theta(i)$ steps to \split{} operations during \RatRace, and with high probability $O(i)$ steps to the remainder of \RatRace.
Hence, by the union bound applied to all $i>0$, the expected maximum number of steps any process needs for \RatRace is asymptotically bounded by the number of steps it devotes to \Imp.

\paragraph{Detailed Analysis.}
First, we modify \Comb such that
%consider a modification of \Comb, called \mComb, where
there is no initial doorway $D$, and processes do not access the 2-process TAS object $T_{top}$ after winning \RatRace or \Imp. %\ref{C1}.
Instead, a process simply terminates if it wins \RatRace or \Imp.
Since the expected \maxstep{} complexity of $T_{top}$ is constant, removing $T_{top}$ does not affect the asymptotic expected \maxstep{} complexity.
%\gtodo{this if the first time we use the term "simulation"; after that we use it a lot; maybe say "algorithm" instead?}
%\pwtodo{``Algorithm'' doesn't seem to fit well. I removed ``simulation'' when possible, and replaced it with ``execution'' when necessary, being aware that term is a bit overloaded now.}
%We also fix a set $P$ of processes that may call \mComb---this set $P$ corresponds to the set of processes that, in a random execution of an algorithm containing \Comb, enter the initial doorway $D$ of \Comb before the first process exits $D$:
We will refer to this modified algorithm as \mComb.

Consider an execution prefix $\EE$ of an algorithm $M$ that uses \Comb, where $\EE$ ends when the first process exits doorway $D$, and suppose $P$ is the set of processes that enter $D$ during $\EE$.
Then the max-contention $k_{\max}^{M,\Comb}(\EE)$ is at least $|P|$.
Hence, it suffices to show for any set $P$, that a random execution of \mComb by the processes in $P$ has expected \maxstep{} complexity $O\bigl(f(|P|)\bigr)$ if scheduled by a location-oblivious (or r/w-oblivious) adversary, and $O(\log|P|)$ if scheduled by a strong adaptive adversary.

To that end, let $M_C$ be the algorithm in which the processes in $P$ (and only them) call \mComb, and let $A_C$ be some adversary.
%Let
%%$\tau_C:\{1,\dots,n\}\to\IIR_{>0}$
%$\tau_C(k)$
%be the expected %adaptive
%\maxstep{} complexity of \mComb against $A_C$.
%As argued above, it suffices to show that
%\begin{align}
%  \tau_C(|P|)=
%  \begin{cases}
%  O\bigl(f(|P|)\bigr), & \text{if $A_C$ is location-oblivious (or r/w-oblivious); and}\\
%  O(\log|P|), & \text{if $A_C$ is strong adaptive.}
%  \end{cases}
%\end{align}
%
A scheduling of $M_C$ by $A_C$ yields a random execution in which a subset of the processes in $P$ take steps (the max-congestion in that execution is $|P|$).
%Let $\Lambda_C[\EE]$ be the sequence of steps in that execution that are devoted to \mComb, and let $\Lambda_I[\EE]$ and $\Lambda_R[\EE]$ denote the sequences of steps in $\EE$ that are devoted to the executions of \Imp and \RatRace, respectively.
%Hence, $\Lambda_C[\EE]$ is an interleaving of $\Lambda_I[\EE]$ and $\Lambda_R[\EE]$.
For two random coin flip vectors $\omega_I,\omega_R\in\Omega^\infty$, let $\EE_C$ denote the random
execution of $M_C$ scheduled by $A_C$, where the $i$-th coin flip result obtained during the execution of \Imp and \RatRace within $\mComb$ is the $i$-th element of $\omega_I$ and $\omega_R$, respectively.
%If $\omega_X$ $\omega_I$, and $\omega_R$ are clear from the context, we simply write $\EE$ instead of $\EE(\omega_I,\omega_R)$.
%(Note that $\EE(\omega_I,\omega_R)$ is still a random execution, if $M_C$ uses randomness outside of \mComb.)
For a process $p\in P$, let $T_C^p$ denote the number of steps $p$ executes in $\EE_C$, and let $T_I^p$ and $T_R^p$ denote the number of those steps that are devoted to \Imp and \RatRace, respectively.
Let $T_C=\max_{p\in P} T_C^p$, $T_I=\max_{p\in P} T_I^p$, and $T_R=\max_{p\in P} T_R^p$.
Then $\Exp[T_C]$ is the expected \maxstep{} complexity of \mComb in $M_C$ against $A_C$.
%Let $T_C(\omega_I,\omega_R)$, $T_I(\omega_I,\omega_R)$ and $T_R(\omega_I,\omega_R)$ denote the maximum number of steps any process executes in $\Lambda_C[\EE(\omega_I,\omega_R)]$, $\Lambda_I[\EE(\omega_I,\omega_R)]$ and $\Lambda_R[\EE(\omega_I,\omega_R)]$, respectively.
%Thus, $T_C(\omega_I,\omega_R)$ equals the \maxstep{} complexity of \mComb in $\EE(\omega_I,\omega_R)$.
%Recall that $A_C$ is an arbitrary adversary, i.e., it may be strong adaptive or weaker than that.
\begin{lemma}\label{lem:combined-algo-analysis}
  There are constants $d_I,d_R>0$ such that
  \begin{align}
    &\Exp[T_C]\leq d_I\cdot (\Exp[T_I]+1),\ \text{and}\label{eq:lem:combined-algo-analysis-a}
    \\
    &\Exp[T_C]\leq d_R\cdot (\Exp[T_R]+1).\     \label{eq:lem:combined-algo-analysis-b}
  \end{align}
\end{lemma}

Before we prove Lemma~\ref{lem:combined-algo-analysis}, we argue that it implies parts~\ref{item:a} and~\ref{item:b} of Theorem~\ref{thm:combine-algos}.

To prove part~\ref{item:a}, we assume that adversary $A_C$ is location-oblivious (or r/w-oblivious).
Let $M_I$ be the algorithm in which the process in $P$ call \Imp.

We construct
%\gtodo{check reviewer's comment marked in response file}
a location-oblivious (or r/w-oblivious) adversary $A_I$ that schedules $M_I$ by simulating adversary $A_C$ as follows.
Let $\omega_R^\ast\in\Omega^\infty$ be a coin flip sequence such that $E[T_I\mid \omega_R=\omega_R^\ast]$ is maximized.
To schedule an execution of algorithm $M_I$, adversary $A_I$ simulates adversary $A_C$ on algorithm $M_C$, using the $i$-th element of $\omega_R^\ast$ for the $i$-th coin flip used in \RatRace.
By the structure of $\mComb$, in which processes alternate steps of \RatRace and \Imp, it is uniquely determined when a process $p$ executes its $i$-th step of \Imp.
Therefore, even the location-oblivious (or r/w-oblivious)
adversary can simulate all steps of \RatRace in $\EE_C$, and schedule processes to take steps in $M_I$ exactly in the same order as they take steps in the \Imp portion of $\EE_C$.

Let $\tau_I$ denote the expected \maxstep{} complexity of \Imp %in $M_I$
against adversary $A_I$.
Then we have $E[T_I\mid \omega_R=\omega_R^\ast]\leq \tau_I(|P|)$.
Since $\omega_R^\ast$ is chosen to maximize the conditional expectation on the left side, it follows that $E[T_I]\leq \tau_I(|P|)$.
Moreover, by the theorem's assumption that the expected \maxstep{} complexity of $\Imp$ is bounded by $f$, we have $\tau_I(|P|)\leq f(|P|)$.
From the last two inequalities and~\eqref{eq:lem:combined-algo-analysis-a}, we obtain $E[T_C]\leq d_I\cdot (f(|P|)+1)=O(f(P))$.
As this is true for all sets $P$, and any location-oblivious (or r/w-oblivious)
adversary $A_C$, it proves part~\ref{item:a} of Theorem~\ref{thm:combine-algos}.
% % % % % % % % % % % % % % % % % % % % % % % % % % % % % % % % % % % % % %
% % DON'T REMOVE THIS COMMENT
% % % % % % % % % % % % % % % % % % % % % % % % % % % % % % % % % % % % % %
% Perhaps a better argument is one by contradiction: Suppose $\tau_C(k)\geq ...$.
% Then there is an execution prefix for which something bad happens.
% We can just take the set P of processes that acces Comb following that prefix, and
% construct an adversary A_I for that specific set (and execution prefix).
% % % % % % % % % % % % % % % % % % % % % % % % % % % % % % % % % % % % % %

The proof of part~\ref{item:b} is almost identical: We now assume $A_C$ is a strong adaptive adversary.
We construct a strong adaptive adversary $A_R$ which schedules processes in $P$ to execute \RatRace by simulating adversary $A_C$ on algorithm $M_C$, assuming the worst-case vector $\omega_I$.
As before, we argue that $E[T_R] \leq \tau_R(|P|)$, where $\tau_R$ is the expected \maxstep{} complexity of \RatRace against $A_R$.
Since by Theorem~\ref{thm:ratrace} the expected \maxstep{} complexity of \RatRace is $O(\log k)$ against any strong adaptive adversary, $\tau_R(|P|)=O(\log |P|)$.
Then from~\eqref{eq:lem:combined-algo-analysis-b}
%at least $\tau_C(|P|)/d_R-1$, where $d_R$ is the constant from Lemma~\ref{lem:combined-algo-analysis}.
%Since by Theorem~\ref{thm:ratrace}  the expected \maxstep{} complexity of \RatRace is $O(\log k)$ against any strong adaptive adversary, this yields $\tau_C(k)=O(\log k)$.
it follows $E[T_C] = O(\log |P|)$.
This completes the proof of Theorem~\ref{thm:combine-algos}.
It remains to prove Lemma~\ref{lem:combined-algo-analysis}.

\subsection{Proof of Lemma~\ref{lem:combined-algo-analysis}}
We first prove~\eqref{eq:lem:combined-algo-analysis-b}.
Consider a process $p\in P$ that invokes \mComb in $\EE_C$.
Process $p$ alternates devoting steps to \Imp and \RatRace (starting with a step of \Imp), until either its \mComb call ends, because $p$ won \RatRace or lost \RatRace or won $\Imp$ (see \ref{C1} and \ref{C2}), or until it stops executing steps of \Imp (see \ref{C3}).
Hence, in either case at least $\floor{T_C^p/2}$ of $p$'s steps in $\EE_C$ are devoted to \RatRace, and thus $T_R\geq (T_C-1)/2$.
This implies \eqref{eq:lem:combined-algo-analysis-b}.

Next we prove \eqref{eq:lem:combined-algo-analysis-a}.
We will use the next statement which follows easily from Chernoff Bounds.
\begin{lemma}
  \label{lem:chernoff_geometric}
  For every constant $0<q<1$, there exists a constant $c>0$ such that the following is true for all $\Delta\geq 0$, and all integers $m\geq 1$.
  If $X_1,\dots,X_m$ are random variables satisfying ${\Pr(X_i> \ell\mid X_1,\dots,X_{i-1})\leq q^{\ell}}$ for every integer $\ell\geq 0$, then
  \begin{displaymath}
  \Pr\left(\sum_{1\leq i\leq m} X_i>c\cdot(m+\Delta)\right)\leq 4^{-\Delta}.
  \end{displaymath}
\end{lemma}
\begin{proof}
  If we choose $c\geq\ceil{\log_4(1/q)}$, then the statement is true for $m=1$.
  Therefore, in the rest of the proof it suffices to consider $m\geq 2$.
  Conditionally on $X_1,\ldots,X_{i-1}$, random variable $X_i$ is dominated by a geometric random variable $Y_i$ with parameter $q$.
  It follows that $\sum_{1\leq i\leq m} X_i$ is dominated by $\sum_{1\leq i\leq m} Y_i$, where the random variables $Y_i$ are mutually independent.
  Let $Y=\sum_{1\leq i\leq m}Y_i$, so $\Exp[Y]=m/q$.
  We can then apply a Chernoff Bound for independent geometric random variables (e.g., \cite[Theorem~1.14]{doerr2011analyzing}), which states that for any $\delta>0$,
  \begin{displaymath}
  \Pr\left(Y\geq (1+\delta)\Exp[Y]\right)\leq\exp\left(-\frac{\delta^2(m-1)}{2(1+\delta)}\right).
  \end{displaymath}
  Setting $\delta=cq+cq\Delta/m-1$, for a $c$ large enough that $\delta>0$, we obtain
  \begin{multline}\label{eq:chernoff_geometric}
  \Pr\left(\sum_{1\leq i\leq m} X_i\geq c\cdot(m+\Delta)\right)
  \leq
  \Pr\bigl(Y\geq c\cdot(m+\Delta)\bigr)
  =
  \Pr\bigl(Y\geq E[Y]\cdot(q/m)\cdot c\cdot(m+\Delta)\bigr)
  \\ \leq
  \Pr\bigl(Y\geq \Exp[Y](1+cq+cq\Delta/m-1)\bigr)
  \leq
  \exp\left(-\frac{(cq+cq\Delta/m-1)^2(m-1)}{2(cq+cq\Delta/m)}\right).
  \end{multline}
  For $c\geq 4/q$, we have $(cq+cq\Delta/m-1)^2\geq (cq+cq\Delta/m)^2/2$, and thus
  \begin{multline*}
  \frac{(cq+cq\Delta/m-1)^2(m-1)}{2(cq+cq\Delta/m)}
  \geq
  \frac{(cq+cq\Delta/m)^2(m-1)}{4(cq+cq\Delta/m)}
  =
  (cq+cq\Delta/m)(m-1)/4
  \\ \geq
  cq(m-1)/4+cq\Delta/8
  >	
  cq\Delta/8.
  \end{multline*}
  (For the second to last inequality we used $m\geq 2$.)
  For large enough $c$, this is at least $\Delta\ln 4$, and then the claim follows from~\eqref{eq:chernoff_geometric}.
\end{proof}

%The next claim states that the probability a process devotes $\Delta$ more steps to \RatRace than to \Imp decreases exponentially with $\Delta$.

The next lemma bounds the probability a process devotes more steps to \RatRace than to \Imp.
Let $Q$ be the set of processes $p\in P$ that stop at some \RatRace splitter when executing \mComb in~$\EE_C$.
%and for each $p\in P$,
%let $Q_p$ be the indicator random variable which is 1 if $p\in Q$ and is 0 otherwise.

\begin{lemma}
  \label{lem:T_R-concentration}
  There is a constant $c>0$ such that for all $\Delta\geq 0$ and any process $p\in P$,
  \begin{displaymath}
  \Pr\bigl(T_R^p>c(T_I^p+\Delta)\mid T_I^p,\,p\in Q\bigr)\leq 4^{-\Delta}.
  \end{displaymath}
\end{lemma}
\begin{proof}
  In the \RatRace portion of \mComb, a process first executes only \split{} operations until it either loses \RatRace (and thus \mComb), or stops at a splitter.
  After stopping at a splitter, $p$'s remaining execution of \RatRace comprises only \TAS{} operations on 2-process and 3-process TAS objects.
  In particular, $p$ executes at most one such \TAS{} operation for each splitter it went through until it stopped at one.
  Recall also that once $p$ has finished the \Imp portion of \mComb, it finishes at most one more \split{} call in \RatRace (if it has a pending such call).
  Hence, $p$ executes at most $T_I^p$ \split{} calls in the \RatRace portion of \mComb, and thus also at most $T_I^p$ \TAS{} operations.
  Thus, defining $Z$ as the number of steps $p$ takes during those \TAS{} operations, we have
  \begin{equation}\label{eq:T_R<=T_TAS}
  T_R^p\leq Z+T_I^p+O(1).
  \end{equation}

  For $i\in\{1,\dots,T_I^p\}$ let $Z_{i}$ denote the number of steps process $p$ executes in order to finish its $i$-th \TAS{} operation on a 2- or 3-process TAS object of the \RatRace portion of \mComb; if $p$ executes fewer than $i$ such \TAS{} operations, then $Z_{i}=0$.
  As discussed in Section~\ref{sec:prelim-objects}, for $\ell\geq 0$ a process finishes a \TAS{} on a 2-process TAS object in $O(\ell)$ steps with probability at least $1-1/2^\ell$.
  We can implement each 3-process TAS object from two 2-process TAS objects in such a way that for each \TAS{} operation on the 3-process TAS, a process needs only to complete one or two \TAS{} operations on the 2-process TAS objects.
  This way, we get the same asymptotic bound as for 2-process TAS objects, i.e., for $\ell\geq 0$, with probability at least $1-1/2^\ell$ a process finishes a \TAS{} operation on a 3-process TAS object in $O(\ell)$ steps.
  Therefore, there is a constant $s>0$ such that $\Pr(Z_i>s\ell\mid T_I^p,Z_1,\dots,Z_{i-1},\,p\in Q)\leq 2^{-\ell}$ for all $\ell\geq 0$.
  Then by Lemma~\ref{lem:chernoff_geometric}, applied to $X_i=Z_i/s$, there is a constant $c'>0$, so that  for all $\Delta\geq 0$ and all $m\geq 1$,
  \begin{displaymath}
  \Pr\bigl(Z>c'(T_I^p+\Delta)\mid T_I^p,\,p\in Q)
  =
  \Pr\bigl (X_1+\dots+X_{T_I^p} >(c'/s)\cdot(T_I^p+\Delta) \mid T_I^p,\,p\in Q\bigr)
  \leq 4^{-\Delta}.
  \end{displaymath}
  Applying  (\ref{eq:T_R<=T_TAS}) yields the claim for a sufficiently large constant $c>0$.
\end{proof}

By Lemma~\ref{lem:T_R-concentration} (used for the inequality labeled $(\ast)$ below), there is a constant $c > 0$ such that for any $\Delta\geq 0$,
\begin{equation}\label{eq:combined-algo-analysis-40}
  \Pr(T_R^p>2c\Delta\mid T_I^p\leq\Delta,p\in Q)
  \leq
  \Pr\bigl(T_R^p>c(T_I^p+\Delta)\mid T_I^p\leq\Delta,p\in Q\bigr)
  \stackrel{(\ast)}{\leq} 4^{-\Delta}.
\end{equation}

Recall that in \RatRace a process can stop either at a randomized splitter on the primary tree, or at a deterministic splitter on an elimination path.
Moreover, at most one process can stop at each splitter, so at most $2^i$ processes can stop at the $i$-th splitter they go through.
Since a process $p$ executes fewer than $T_I^p$ \split{} calls in \RatRace before stopping at a splitter or terminating \mComb, the number of processes $p\in Q$ satisfying $T_I^p\leq\Delta$ is at most $2^\Delta$.
Hence,
\begin{equation}\label{eq:combined-algo-analysis-42}
  \sum_{p\in P}\Pr(T_I^p\leq \Delta \wedge p\in Q)
  =
  \Exp[|\{p\in Q: T_I^p\leq\Delta\}|]
  \leq
  2^\Delta.
\end{equation}
For any process $p\in P$ that does not stop at any \RatRace splitter, i.e., $p\in P\setminus Q$, we have $T_R^p \leq T_I^p + O(1)$, because once $p$ has finished the \Imp portion of \mComb, it finishes at most one \split{} call in \RatRace before finishing \mComb.
It follows that for any $p\in P$, $T_R^p>2cT_i^p$ implies $p\in Q$, if the constant $c$ is sufficiently large.
Using this observation we obtain
\begin{multline}\label{eq:combined-algo-analysis-45}
  \sum_{p\in P} \Pr(T_R^p>2c\Delta\wedge T_I^p\leq\Delta)
  =
  \sum_{p\in P} \Pr(T_R^p>2c\Delta\wedge T_I^p\leq\Delta\wedge p\in Q)
  \\=
  \sum_{p\in P} \Pr(T_R^p>2c\Delta \mid T_I^p\leq\Delta, p\in Q)\cdot\Pr(T_I^p\leq\Delta \wedge p\in Q)
  \\ \stackrel{(\ref{eq:combined-algo-analysis-40})}\leq
  \sum_{p\in P}4^{-\Delta}\cdot\Pr(T_I^p\leq\Delta \wedge p\in Q)
  \stackrel{(\ref{eq:combined-algo-analysis-42})}{\leq}
  4^{-\Delta}\cdot 2^{\Delta}
  =
  2^{-\Delta}.
\end{multline}
It follows that
\begin{align}\label{eq:combined-algo-analysis-47}
  \Pr(T_R>2c\Delta)
  \notag&=
%  \Pr\bigl((T_R>2c\Delta \wedge T_I> \Delta) \vee (T_R>2c\Delta \wedge T_I\leq\Delta)\bigr)
%  \\ \notag &\leq
  \Pr(T_R>2c\Delta \wedge T_I>\Delta)+\Pr(T_R>2c\Delta \wedge T_I\leq\Delta)
  \\ \notag &\leq
%  \Pr(T_I>\Delta)+\sum_{p\in P} \Pr(T_R^p>2c\Delta \wedge T_I\leq\Delta)
%  \\ \notag &\leq
  \Pr(T_I>\Delta)+\sum_{p\in P} \Pr(T_R^p>2c\Delta \wedge T_I^p\leq\Delta)
  \\ &\stackrel{\ref{eq:combined-algo-analysis-45}}{\leq}
  \Pr(T_I>\Delta)+2^{-\Delta}.
%  =
%  \Pr(T_I>t/(2c))+2^{-t/(2c)}.
\end{align}
Then
\begin{multline*}
  \Exp[T_R]
  =
  \sum_{t\geq 0}\Pr(T_R>t)
  \stackrel{(\ref{eq:combined-algo-analysis-47})}{\leq}
  \sum_{t\geq 0}\Bigl(\Pr\bigl(T_I>t/(2c)\bigr)+2^{-t/(2c)}\Bigr)
  \\ \leq
  \sum_{t\geq 0}\Pr\bigl(T_I>\floor{t/(2c)}\bigr)+\sum_{t\geq 0}2^{-t/(2c)}
  \leq
  \sum_{j\geq 0}2c\cdot\Pr(T_I>j)+O(1)
  =
  O(\Exp(T_I)).
\end{multline*}
Finally, combining that with the fact that $T_C=T_I+T_R$, implies~\eqref{eq:lem:combined-algo-analysis-a}.
%Now (\ref{eq:lem:combined-algo-analysis-a}) follows from the upper bound on $\Exp[T_R]$ we have just shown, and the fact that $T_C=T_I+T_R$.
This completes the proof of Lemma~\ref{lem:combined-algo-analysis}.
\qed

\section{A 2-Process Time Lower Bound for Oblivious Adversaries}
\label{sec:lb-two-processes}

%Recall that $\EE_{M,A}$ denotes a random execution of algorithm $M$ scheduled by adversary $A$.
%%obtained by adversary $A$ for algorithm $M$.
%In this section we prove a lower bound for the \maxstep{} complexity of \TAS{} in $\EE_{M,A}$ for any 2-process TAS algorithm $M$, against the ``worst'' oblivious adversary $A$.

We show a lower bound on the \maxstep{} complexity of any 2-process TAS implementation, % $M$,
against the worst possible oblivious adversary. %$A$.
%Recall that $\EE_{M,A}$ denotes a random execution of $M$ scheduled by $A$.

\begin{theorem}
    \label{thm:lb2TAS}
    For any randomized 2-process TAS implementation % $M$
    and any integer $t\geq 0$, there is an oblivious adversary $A$ such that
%    with probability at least $1/4^t$ the \maxstep{} complexity of the implemented \TAS{} operation in
%%    $\EE_{M,A}$
%    a random execution
%%    of $M$
%    scheduled by $A$
%    is at least $t$.
    with probability at least $1/4^t$ the \maxstep{} complexity of the implemented \TAS{} operation against $A$ is at least $t$.
\end{theorem}
\begin{proof}
    The proof employs Yao's minimax principle~\cite{Yao1977}.
%    Yao's Min-Max Lemma~\cite{Yao1977}.

%    Let $\MM_t$ denote the set of all \emph{deterministic} implementations of a method $\FuncSty{m()}$ for two processes, $p_0$ and $p_1$, in which each process' method call responds either with a return value in $\{0,1\}$, or with no return value, and in all executions the \maxstep{} complexity of method \FuncSty{m()} is bounded by $t$.
%    Formally, processes calling $\FuncSty{m()}$ still alternate between coin flips and
%
%    For any implementation $M\in \MM_t$ and any schedule $\sigma$ over $\{p_0,p_1\}$, let $\EE_M(\sigma)$ denote the execution obtained if processes $p_0$ and $p_1$ execute $M$ according to schedule $\sigma$.
%    Further, let $c(M,\sigma)=0$ if all $\FuncSty{m()}$ calls in $\EE_M(\sigma)$ respond, and
%    execution $\EE_M(\sigma)$ is linearizable with respect to the specification of TAS (i.e., if \FuncSty{m()} were a \TAS{} method).
%    Otherwise, $c(M,\sigma)=1$.

    Let $M$ be a %linearizable
    randomized implementation of a 2-processes TAS object.
    For any execution $\EE$
    of this implementation, let $c_t(\EE)=1$ if
%    some of the \TAS{} calls invoked in $\EE$ does not respond in $\EE$;
%    none of the \TAS{} calls invoked in $\EE$ responds in $\EE$;
    some process executes at least $t$ shared memory steps in $\EE$;
%    (before finishing its \TAS{} call);
    let $c_t(\EE)=0$ otherwise.

%    Further, let $M_t$ be a \TAS{} implementation, where each process executes exactly the same steps as in $M$, until its \TAS{} call responds, or it has executed its $t$-th shared memory step (whichever comes first).
%    If by the $t$-th shared memory step the process' \TAS{} call has not responded, then that call immediately responds without a return value.

    Let $\Sigma_t$ be the set of all possible schedules
    $\sigma = (\sigma_1,\sigma_2,\ldots)$, where $\sigma_i\in\{0,1\}$, and $\sigma$ has the following properties:
    (i)~$|\sigma| = 2k$, for some $k\in\{t,\ldots,2t-1\}$;
    (ii)~$\sigma_{2i-1}=\sigma_{2i}$, for all $i\in\{1,\dots,k\}$; and
    (iii)~some process $p\in\{0,1\}$ appears exactly $2t$ times at $\sigma$ (so, the other process, $1-p$, appears $2(k-t)<2t$ times).
%    $\sigma=(\sigma_1,\sigma_2,\dots,\sigma_{4t})\in\{0,1\}^{4t}$, where $\sigma_{2i-1}=\sigma_{2i}$ for any $i\in\{1,\dots,2t\}$, and each process appears exactly $2t$ times.
    We have
    \begin{equation}\label{eq:2-process-lowerbound-10}
    |\Sigma_t|
    \leq
    \sum_{k=t}^{2t-1}2^{k}
    \leq
    2^{2t}=4^t.
    \end{equation}

    Consider the coin flip sequences $\omega_p=(\omega_{p,1},\dots,\omega_{p,t})\in\Omega^{t}$, for $p\in\{0,1\}$.
    For any schedule $\sigma\in\Sigma_t$, let $\EE_{M}(\sigma,\omega_0,\omega_1)$ denote the execution of algorithm $M$ where processes are scheduled according to $\sigma$, and the $i$-th coin flip of process $p$ returns the value $\omega_{p,i}$.
    Recall our model assumption that (w.l.o.g.) each process alternates between coin flip steps and shared memory steps.
    Since each process appears at most $2t$ times in $\sigma\in\Sigma_t$, each process executes at most $t$ coin flips in the resulting execution.
    We will now show that
    \begin{equation}\label{eq:2-process-lowerbound-20}
      \forall\, \omega_0,\omega_1\in\Omega^{t}\,\exists\, \sigma\in\Sigma_t\colon  c_t\bigl(\EE_{M}(\sigma,\omega_0,\omega_1)\bigr)=1.
    \end{equation}

    To prove \eqref{eq:2-process-lowerbound-20}, let
%    \gtodo{already defined above}
%    $\omega_p=(\omega_{p,1},\dots,\omega_{p,t})\in\Omega^t$, for $p\in\{0,1\}$.
%    Let
    $\lambda_p$, for $p\in\{0,1\}$, be an arbitrary but fixed infinite extension of $\omega_p$, e.g., we can choose $\lambda_p=(\omega_{p,1},\dots,\omega_{p,t},0,0,\dots)$, assuming that $0$ is an element of $\Omega$.
    Let $M_{\lambda}$ be the TAS algorithm where each process $p$ executes the same program as in $M$, but ignores its coin flips, and instead acts as if its $i$-th coin flip is %$\lambda_{p,i}$.
    the $i$-th element of vector $\lambda_p$.
    Then $M_{\lambda}$ behaves as a deterministic 2-process TAS algorithm.
    Since there is no \emph{wait-free} deterministic 2-process TAS algorithm,
    there exists an execution of $M_{\lambda}$ in which at least one process executes at least $t$ shared memory steps without finishing its \TAS{} call.
    Moreover, there is such an execution $\EE'$ which has the additional property that each coin flip step by process $p\in\{0,1\}$ (whose result is
%    being ignored by the algorithm)
    replaced in the algorithm by an element of $\lambda_p$)
    is immediately followed in $\EE'$ by the next shared memory step of the same process $p$.
    Let $\EE$ be the prefix of $\EE'$ that ends when the first process has executed its $t$-th shared memory step, and let $\sigma$ be the schedule corresponding to $\EE$.
%    Then at least one \TAS{} call in $\EE$ does not respond, so $c_t(\EE)=1$.
    The prefix $\EE$ exists, because we argued above that some process executes at least $t$ shared memory steps without finishing its \TAS{} call.
    It follows that $c_t(\EE)=1$ and $\sigma\in\Sigma_t$, and also $\EE=\EE_M(\sigma,\omega_0,\omega_1)$.
%    Moreover,
%%    since in $\EE$ no process executes more than $t$ shared memory steps, and each coin flip by a process $p$ is immediately followed by a shared memory step by $p$, $\sigma\in\Sigma_t$.
%    it follows from the definition of $\EE$ that the schedule $\sigma$ belongs to $\Sigma_t$, and also $\EE=\EE_M(\sigma,\omega_0,\omega_1)$.
%%    , because each process executes at most $t$ coin flips in $\EE$ and $\omega_p$ is the prefix of length $t$ of $\lambda$.
    This proves \eqref{eq:2-process-lowerbound-20}.

    Now let $(\omega_0^\ast,\omega_1^\ast)$ be chosen according to any product distribution over $\Omega^t\times\Omega^t$, and $\sigma^\ast$ according to any probability distribution over $\Sigma_t$.
    By
%    Yao's Min-Max Lemma~\cite{Yao1977},
    Yao's minimax principle~\cite{Yao1977},    \begin{equation}\label{eq:2-process-lowerbound-30}
%    \varepsilon:=
    \max_{\sigma\in\Sigma_t} \Exp\left[c_t\bigl(\EE_M(\sigma,\omega_0^\ast,\omega_1^\ast)\bigr)\right]
    \geq
    \min_{\omega_0,\omega_1\in\Omega^t} \Exp\left[c_t\bigl(\EE_M(\sigma^\ast,\omega_0,\omega_1)\bigr)\right].
    \end{equation}
    Let
    $\varepsilon$ denote the left side in this inequality, and
    recall that $c_t (\EE_M(\sigma,\omega_0^\ast,\omega_1^\ast))$ is a 0--1 random variable indicating whether some process executes at least $t$ steps in $\EE_M(\sigma,\omega_0^\ast,\omega_1^\ast)$.
    %before finishing its \TAS{} call.
    Hence, $\varepsilon$ is a lower bound for the probability that some process needs at least $t$ steps to finish its \TAS{} call in a random execution of $M$, for the worst possible schedule~$\sigma$.
    Thus, it suffices to prove that $\varepsilon\geq 1/4^t$.
    To do so we choose $\sigma^\ast$ uniformly in $\Sigma_t$ to obtain
    \begin{displaymath}
      \varepsilon
      \stackrel{\eqref{eq:2-process-lowerbound-30}}{\geq}
      \min_{\omega_0,\omega_1\in\Omega^t} \Exp\left[c_t\bigl(\EE_M(\sigma^\ast,\omega_0,\omega_1)\bigr)\right]
      =
      \min_{\omega_0,\omega_1\in\Omega^t} \Pr\Bigl(c_t\bigl(\EE_M(\sigma^\ast,\omega_0,\omega_1)\bigr)=1\Bigr)
      \stackrel{\eqref{eq:2-process-lowerbound-20}}{\geq}
      \frac{1}{|\Sigma_t|}
      \stackrel{\eqref{eq:2-process-lowerbound-10}}{\geq}
      \frac{1}{4^t}.
    \end{displaymath}
    This completes the proof of Theorem~\ref{thm:lb2TAS}.
\end{proof}

\section*{Conclusion}

In this paper  we devised several efficient randomized TAS algorithms.
Most importantly, we presented an algorithm with an expected \maxstep{} complexity of $O(\log^\ast k)$
against the oblivious and some slightly stronger adversary models,
where $k$ is a measure of contention.

The progress in improving randomized TAS algorithms is mirrored by recent progress on randomized consensus algorithms.
Aspnes~\cite{Aspnes2012} has devised a randomized consensus algorithm that has $O(\log\log n)$ expected \maxstep{} complexity in the oblivious adversary model.
This algorithm is based on the sifting technique from~\cite{AA_TAS_2011a}.
It would be interesting to investigate whether techniques similar to those presented here can be used to achieve even faster consensus algorithms.
In particular, we believe that our group election implementation for r/w-oblivious adversaries proposed in Section~\ref{sec:rw-algo} could be used in the framework of~\cite{Aspnes2012} to obtain an \emph{adaptive} binary consensus algorithm with expected \maxstep{} complexity $O(\log\log k)$.

Several other important problems remain open.
For the oblivious adversary, no TAS implementations with constant expected \maxstep{} complexity are known, and no super-constant lower bounds are known even in the strong adaptive adversary model.
% Further, there is an exponential gap between the known lower bound of $\Omega(\log n)$ and the upper bound of $O(n)$ for the number of registers that are required for randomized or obstruction-free TAS algorithms.

\section*{Acknowledgements}
We thank Dan Alistarh for pointing out Styer and Peterson's $\Omega(\log n)$  space lower bound for deadlock-free leader election \cite{Styer1989}.
We also thank the anonymous reviewers for their helpful feedback.

%\bibliography{tas}

%\appendix
%
%\section{Appendix: A Chernoff Bound}
%
%The \gtodo{move to main text?}
%following Chernoff Bound can be found in~\cite[Theorem~2.3(b)]{McDiarmid1998}.
%%In Section~\ref{sec:Rat_Race_Space_Improvement} we use a Chernoff Bound in the form presented by McDiarmid \cite{McDiarmid1998}, Theorem~2.3\,(b).
%\begin{theorem}[Chernoff Bound]
%\label{thm:Chernoff}
%  Let $X_1,X_2,\dots,X_n$ be independent random variables with $0\leq X_i\leq 1$, for each $i\in\{1,\dots,n\}$, and let $X=X_1+\dots+X_m$ and $\mu=\Exp[X]$.
%  Then for any $\delta>0$,
%  \begin{displaymath}
%    \Pr\bigl(X\geq (1+\delta)\mu\bigr)\leq e^{-\frac{\delta^2\mu}{2(1+\delta/3)}}.
%  \end{displaymath}
%\end{theorem}

\end{document}